\documentclass[10pt,journal,compsoc]{IEEEtran-jrnl}

\usepackage{cite}
\usepackage{amssymb,amsmath,amsthm}

\usepackage{array}
\usepackage{graphicx}
\usepackage{amsmath}
\usepackage{amssymb}
\usepackage{algorithm}
\usepackage[noend]{algorithmic}
\usepackage{tikz}
\usepackage{xcolor}
\usepackage{multirow}
\usepackage{paralist}
\usepackage{url}
\usepackage{float}
\usepackage{makecell}
\usepackage{mathtools}
\usepackage{marvosym}
\usepackage{etoolbox}

\providebool{ArxivVersion}
\setbool{ArxivVersion}{true}

\usetikzlibrary{calc,shadows,patterns,shapes,arrows,decorations.pathmorphing,backgrounds,positioning,fit,plotmarks}

\newtheorem{lemma}{Lemma}
\newtheorem{theorem}{Theorem}

\newtheorem{definition}{Definition}

\newcommand{\setof}[1]{\left\{{#1}\right\}}
\newcommand{\set}[2]{\left\{#1\mid #2\right\}}

\newcommand{\tsocs}[0]{\mbox{\emph{MMSS}}}

\newcommand{\litmus}[0]{$\text{LITMUS}^{\text{RT}}$}

\usepackage{array}
\newcolumntype{L}[1]{>{\raggedright\let\newline\\\arraybackslash\hspace{0pt}}m{#1}}
\newcolumntype{C}[1]{>{\centering\let\newline\\\arraybackslash\hspace{0pt}}m{#1}}
\newcolumntype{R}[1]{>{\raggedleft\let\newline\\\arraybackslash\hspace{0pt}}m{#1}}

\tikzset{
	task/.style={shade, shading=radial, rectangle,minimum height=.1cm,
		inner color=#1!20, outer color=#1!60!gray},
	task1/.style={task=yellow, minimum width=13mm},
	task2/.style={task=orange, minimum width=13mm},
	task3/.style={task=red, minimum width=13mm},
	task4/.style={task=green, minimum width=13mm},
	task5/.style={task=blue, minimum width=13mm},
	task6/.style={task=purple, minimum width=13mm},
	task7/.style={task=cyan, minimum width=13mm},
	task8/.style={task=pink, minimum width=13mm},
}

\floatname{algorithm}{Algorithm}
\newcommand{\algorithmicinput}{\textbf{Input:}}
\newcommand{\INPUT}{\item[\algorithmicinput]}

\AtBeginDocument{%
	\providecommand\BibTeX{{%
			\normalfont B\kern-0.5em{\scshape i\kern-0.25em b}\kern-0.8em\TeX}}}

\begin{document}
\title{Scheduling of Real-Time Tasks with Multiple Critical Sections in
	Multiprocessor Systems}

\author{Jian-Jia Chen, Junjie Shi, Georg von der Br\"uggen, and Niklas Ueter\\
	Department of Informatics, TU Dortmund University, Germany\\
	\{jian-jia.chen, junjie.shi, georg.von-der-brueggen, niklas.ueter\}@tu-dortmund.de\\

}

%
%

\ifArxivVersion
    \markboth{Early Announcement}%
{Shell \MakeLowercase{\textit{et al.}}: Bare Demo of IEEEtran.cls for Computer Society Journals}
\else
    \markboth{submitted to IEEE Transactions on Computers}%
{Shell \MakeLowercase{\textit{et al.}}: Bare Demo of IEEEtran.cls for Computer Society Journals}
\fi

%



\IEEEtitleabstractindextext{%
\begin{abstract}
The performance of 
multiprocessor synchronization and locking protocols
is a key factor to utilize the computation power
of multiprocessor systems under real-time constraints. While 
multiple protocols have been developed in the past decades, 
their performance highly depends on 
the task partition and prioritization.
The recently proposed Dependency Graph Approach showed its advantages and
attracted a lot of interest. 
It is, however, restricted to
task sets where each task 
has at most one critical section.
In this paper, we remove this restriction 
and demonstrate how to utilize 
algorithms for the classical job shop
scheduling problem to construct a dependency graph for tasks with multiple
critical sections. 
To show  the 
applicability, we discuss the
implementation in \litmus{} and report the overheads. 
Moreover, we provide 
extensive numerical evaluations under different configurations, 
which  in many situations show significant improvement compared to the
state-of-the-art.
\end{abstract}

}

\maketitle

\IEEEdisplaynontitleabstractindextext

%
\IEEEpeerreviewmaketitle

\IEEEraisesectionheading{\section{Introduction}\label{sec:introduction}}
\IEEEPARstart{U}{nder} 
the von-Neumann programming model, shared resources that require
mutual exclusive accesses, such as shared files, data structures,
etc., have to be protected by applying synchronization (\emph{binary
	semaphores}) or locking (\emph{mutex locks}) mechanisms. A protected
code segment that has to access a shared resource mutually exclusively is
called a \emph{critical section}.  For uniprocessor real-time systems,
the state-of-the-art are longstanding  protocols that have been developed in the
90s,
namely the
Priority Inheritance Protocol (PIP) and the Priority Ceiling Protocol~(PCP) 
by Sha~et~al.~\cite{DBLP:journals/tc/ShaRL90}, as well as the Stack
Resource Policy~(SRP) by Baker~\cite{DBLP:journals/rts/Baker91}.
Specifically, a variant of PCP has been implemented in Ada (called
Ceiling locking) and in POSIX (called Priority Protect Protocol).

Due to the development of multiprocessor platforms,
multiprocessor resource synchronization and locking protocols have
been proposed and extensively studied, such as the Distributed 
PCP
(DPCP)~\cite{DBLP:conf/rtss/RajkumarSL88}, the
Multiprocessor 
PCP (MPCP)~\cite{Rajkumar_1990},
the Multiprocessor SRP 
(MSRP)~\cite{DBLP:conf/rtss/GaiLN01}, the Flexible Multiprocessor
Locking Protocol (FMLP)~\cite{block-2007}, the Multiprocessor
PIP~\cite{DBLP:conf/rtss/EaswaranA09}, the $O(m)$ Locking Protocol
(OMLP)~\cite{DBLP:conf/rtss/BrandenburgA10}, the Multiprocessor
Bandwidth Inheritance (M-BWI)~\cite{DBLP:conf/ecrts/FaggioliLC10},
and the
Multiprocessor resource sharing Protocol
(MrsP)~\cite{DBLP:conf/ecrts/BurnsW13}.  Since the performance of
these protocols highly depends on task partitioning, several
partitioning algorithms were developed in the literature, e.g., for
MPCP by Lakshmanan~et~al.~\cite{DBLP:conf/rtss/LakshmananNR09} and
Nemati~et~al.~\cite{DBLP:conf/opodis/NematiNB10}, for MSRP by
Wieder~and~Brandenburg~\cite{DBLP:conf/sies/WiederB13}, and for DPCP
by Hsiu~et~al.~\cite{DBLP:conf/emsoft/HsiuLK11}, Huang et.
al~\cite{RTSS2016-resource}, and von der Br\"uggen et
al.~\cite{RTNS17-resource}.  In addition to the theoretical soundness
of these protocols, some of them have been implemented in the
real-time operating systems
\litmus~\cite{calandrino2006litmus,bbb-2011} and RTEMS~\footnote{http://www.rtems.org/}.


For several decades, the primary focus when considering multiprocessor
synchronization and locking in real-time systems has been the design and analysis of
resource sharing protocols, 
where the protocols decide the order in which the new incoming requests access
the shared resources dynamically. Contrarily, the Dependency Graph
Approaches (DGA), that was proposed by Chen et al.~\cite{Chen-Dependency-RTSS18} in 2018, pre-computes
the order in which tasks are allowed to access resources, 
and consists of two individual steps: 
\begin{compactenum}
	\item A dependency graph is constructed to determine the
	execution order of the critical sections guarded by \emph{one binary
		semaphore or mutex lock}.
	\item Multiprocessor scheduling algorithms are applied to
	schedule the tasks by respecting the constraints given by the constructed
	dependency graph(s).
\end{compactenum}
Chen~et~al.~\cite{Chen-Dependency-RTSS18} showed significant
improvement against existing protocol-based approaches from the
empirical as well as from the theoretical perspective, and demonstrated 
the practical applicability of the DGA by implementing it in  
\litmus~\cite{calandrino2006litmus,bbb-2011}. However, 
the original dependency graph approaches
presented in~\cite{Chen-Dependency-RTSS18} has two strong limitations: 
1)~the construction in the first step allows only one critical section per task, and
2)~the presented algorithms can only be applied for frame-based
real-time task systems, i.e., all tasks have the same period and
release their jobs always at the same time.
The latter has been
recently 
removed by Shi et~al.~\cite{Shi-Dependency-RTAS2019}, 
who applied the DGA after unrolling the jobs in the
hyperperiod. However, the former remains open and is a fundamental
obstacle which limits the generality of the dependency graph
approaches.

In the original DGA, the assumption that each task has only one non-nested
critical section allows the algorithm to partition the tasks
according to their shared resources in the first step. However, when a
task accesses multiple shared resources, such a partitioning is no
longer possible. Therefore, to enable the DGA for tasks with multiple
critical sections, an exploration of effective construction mechanisms for a
dependency graph that  considers the interactions of the shared
resources is needed.

\noindent\textbf{Contribution:}
In this paper, we focus on
allowing multiple critical sections per task 
in the dependency graph approaches 
for both frame-based 
and periodic real-time
task systems with synchronous releases. 
Our contributions are: 
\begin{compactitem}
	
	\item Our key observation is the correlation between the dependency graph in DGA and
	the classical \emph{job shop scheduling problem}.
	With respect to the 
	computational complexity, we present a polynomial-time reduction from
	the classical \emph{job shop scheduling problem}, which is ${\mathcal NP}$-hard in the
	strong sense~\cite{LenstraRinnooy-Kan:79:Computational-complexity}.
	Intractability results are established even
	for severely restricted instances of the studied multiprocessor
	synchronization problem, as detailed in 
	Sec.~\ref{sec:computational-complexity}.
	
	\item For frame-based task sets, we reduce the problem of constructing the 
	dependency graph in the DGA to the classical \emph{job shop scheduling problem}
	in Sec.~\ref{sec:our-algorithm-based-on-shop}, and  
	establish approximation bounds for minimizing the makespan based on
	the approximation bounds of 
	job-shop algorithms. 
	Sec.~\ref{sec:extention-to-periodic} 
	details how these results can be extended to 
	periodic real-time task systems.	
	
	\item We explain how we implemented the dependency graph approach 
	with multiple critical sections 
	in \litmus{} and report the overheads in Sec.~\ref{sec:implementation-overhead},
	showing that 
	our new implemented approach is comparable to the existing methods with respect to the overheads.
	\item We provide extensive numerical evaluations in
	Sec.~\ref{sec:evaluations}, which demonstrate the
	performance of the proposed approach under different system
	configurations. 
	Compared to the state-of-the-art, our approach shows
	significant improvement for all the evaluated frame-based real-time task systems
	and for most of the evaluated periodic task systems.  
\end{compactitem}

%
%
%
%

\section{System Model}
\label{sec:system-model}

\subsection{Task Model}
\label{sec:task-model}

We consider a set $\bf T$ of $n$ recurrent tasks to be scheduled on $M$
identical (homogeneous) processors. All tasks can have multiple (non-nested) critical
sections 
and may access several 
of the $Z$ shared resources. 
Each task $\tau_i$ 
is described by
\mbox{$\tau_i=((\eta_i, C_i),~T_i,~D_i)$}, where:

\begin{compactitem}
	\item $\eta_i$ is the number of computation segments in task $\tau_i$.
	\item $C_i$ is the total worst-case execution time (WCET) of 
	the computation segments in task $\tau_i$. 
	\item $T_i$ is the period of $\tau_i$.
	\item $D_i$ is the relative deadline of $\tau_i$. 
\end{compactitem}
We consider constrained deadlines,  
i.e., $\forall \tau_i \in \textbf{T},\ D_i \leq T_i$.
For the $j$-th segment of task $\tau_i$, denoted as
$\theta_{i,j} = (C_{i, j}, \lambda_{i,j})$:
\begin{compactitem}
	\item $C_{i,j}\geq 0$ is the WCET of computation segment $\theta_{i,j}$ with  
	$C_i=  \sum_{j=1}^{\eta_i} C_{i,j}$.
	\item $\lambda_{i,j}$ indicates whether the corresponding segment 
	is 
	a non-critical section or a critical section. 
	If $\theta_{i,j}$ is a critical section, $\lambda_{i,j}$ is $1$;
	otherwise, $\lambda_{i,j}$ is $0$.
	\item If $\theta_{i,j}$ is a non-critical section, then $\theta_{i,j-1}$ and
	$\theta_{i,j+1}$ must be critical sections (if they exist). That is,
	$\theta_{i,j}$ and $\theta_{i,j+1}$ cannot be both non-critical
	sections.
	\item If $\theta_{i,j}$ is a critical section, it starts from the lock of a
	mutex lock (or \emph{wait} for a binary semaphore), denoted by $\sigma_{i,j}$, and ends at the unlock
	of the same mutex lock (or \emph{signal} to the same binary semaphore).
\end{compactitem}
\noindent Furthermore, we make following assumptions: 
\begin{compactitem}
	\item A job cannot be executed in parallel, i.e., the computation segments in a
	job must be sequentially executed.
	\item The execution of the critical sections guarded by a mutex lock
	(or one binary semaphore) must be sequentially executed.
	Hence, if two computation segments share the same
	lock, they must be executed one after another.
	\item There are in total $Z$ mutex locks (or binary semaphores). 
\end{compactitem}
\noindent We consider two kinds of task systems, namely:
\begin{compactitem} 
	\item \textbf{Frame-based} task systems: all tasks release
	their jobs at the same time and have the same period and relative deadline, i.e., $\forall i, j,~ T_i = T_j
	\land D_i=D_j$. Hence, the analysis can be restricted to one job of each task.
	\item \textbf{Periodic} task systems (with synchronous release): all tasks
	release their first job at time $0$ and subsequent jobs are released periodically,
	but
	different tasks may have different periods and relative deadlines.
	The hyper-period of the task
	set $\textbf{T}$ is defined as the least common multiple (LCM) of the
	periods of the tasks in~$\textbf{T}$.  
\end{compactitem}

\subsection{Problem Definition and Approximation}
\label{sec:problem-definition}

In this subsection, 
we define the problem of scheduling frame-based
real-time tasks with multiple critical sections in homogeneous
multiprocessor systems. 

We define a schedule from the sub-job's perspective. Suppose that ${\bf
	\Theta}$ is the set of the computation segments, i.e., ${\bf
	\Theta}=\set{\theta_{i,j}}{\tau_i \in \textbf{T},
	j=1,2,\ldots,\eta_i}$. A schedule for $\textbf{T}$ is a function
$\rho: {\mathbb R}\times M \rightarrow \bf{\Theta}\cup\setof{\bot}$,
where $\rho(t,m) = \theta_{i,j}$ denotes that the sub-job
$\theta_{i,j}$ 
is executed at time~$t$ on processor
$m$, and $\rho(t, m) = \bot$ denotes that processor $m$ is idle at
time $t$. Since a job has to be sequentially executed, at any time
point $t \geq 0$, only a sub-job of $\tau_i$ can be executed on one of
the $M$ processors, i.e., if $\rho(t,m)$ is $\theta_{i,j}$,
then $\rho(t,m') \neq \theta_{i,k}$ for any
$k \leq \eta_i$ 
and $m' \neq m$. Moreover, since the sub-jobs of
a job must be executed sequentially, $\theta_{i,k}$ cannot be executed
before $\theta_{i,j}$ finishes for any 
$j<k\leq \eta_i$,
i.e., if $\rho(t,m)$ is $\theta_{i,j}$ for some $t, m, i, j$, then
$\rho(t',m) \neq \theta_{i,k}$ for any $t' \leq t$ and  any $k>j$.
The critical sections guarded by one mutex lock must be sequentially
executed, i.e., if $\lambda_{i,j}$ is $1$, $\lambda_{k,\ell}$ is $1$,
and $\sigma_{i,j}=\sigma_{k,\ell}$,
then when
$\rho(t,m)$ is $\theta_{i,j}$, and a schedule must guarantee that 
$\rho(t,m') \neq \theta_{k,\ell}$ for any $t \geq 0$ and $m \neq m'$.

We only consider schedules that can finish the execution demand of the
computation segments. 
Let $R$ be the finishing time of the
schedule. 
In this case,  $\sum_{m=1}^{M} \int_{0}^{R} [ \rho(t, m) =
\theta_{i,j}] dt$ must be equal to $C_{i,j}$, where $[P]$ is the
Iverson bracket, i.e., $[P]$ is $1$ when the condition $P$ holds, otherwise
$[P]$ is $0$. Note that the integration is used in this
paper only as a symbolic notation to represent the summation over
time. 
The earliest moment when all  sub-jobs finish their computation segments in the schedule (under all the
constraints defined above) is called the \emph{makespan} of the
schedule, commonly denoted as $C_{\max}$ in scheduling theory, i.e.,
$C_{\max}$ of schedule $\rho$ is:
\[
\text{min. } R \text{ ~~s. t. }  \sum_{m=1}^{M} \int_{0}^{R} [
\rho(t, m) = 
\theta_{i,j}] dt = C_{i,j}, \forall \theta_{i,j} \in {\bf \Theta}
\]

The problem of multiprocessor synchronization  with multiple critical
sections per task can be transferred to the following two general problems: 

\begin{definition}
	\label{def:problem-def}
	\textbf{Multiprocessor Multiple critical-Sections task Synchronization
		(\tsocs{}) makespan problem:} 
	Assume $M$ identical (homogeneous) processors and that 
	$n$ tasks are arriving at time $0$. Each task $\tau_i$ is composed of
	$\eta_i$ computation segments, each of which is either a non-nested
	critical section or a non-critical section. The objective is to find
	a schedule that minimizes the makespan.
\end{definition}

A feasible schedule of the \tsocs{} makespan problem is a schedule
that satisfies all aforementioned non-overlapping constraints. 
An optimal solution of an input
instance of the \tsocs{} makespan problem is the makespan of a
schedule that has the minimum makespan among the feasible schedules of
the input instance.  An algorithm ${\mathcal A}$ for the \tsocs{} makespan
problem has an \emph{approximation ratio} $a\geq 1$, if given any
task set $\textbf{T}$ and $M$ processors, the resulting makespan is at
most $a \cdot C_{\max}^*$, where $C_{\max}^*$ is the optimal makespan.

\begin{definition}
	\label{def:problem-decision-version-def}
	\textbf{The \tsocs{}  schedulability problem:} 
	Assume there are $M$ identical (homogeneous) processors and that 
	$n$ tasks are arriving at time $0$. 
	All tasks
	$\tau_i$ have the same deadline~$D$. Each task is composed of $\eta_i$
	computation segments, each of which is either a non-nested critical section or a
	non-critical section.  The objective is to find a feasible schedule that meets
	the deadline $D$ 
	on the given $M$ processors.
\end{definition}

A feasible schedule of the \tsocs{} schedulability problem is a
schedule that has a makespan no more than $D$ and 
satisfies all the non-overlapping constraints. 
The \tsocs{} schedulability problem is a decision problem, in
which for a given $D$ and a given
algorithm either a feasible schedule is derived  that
meets the deadlines or no feasible schedule can be derived from the algorithm.
For such a decision setting, the \emph{speedup
	factor}~\cite{Kalyanasundaram:2000,Phillips:stoc97} can be used to examine
the performance.
\emph{Provided that there exists one feasible schedule at the original
	speed}, the speedup factor \mbox{$a \geq 1$} of a scheduling
algorithm~${\mathcal A}$ for the \tsocs{} schedulability problem is the 
factor $a\geq1$ by which the overall speed of a system would need to be
increased so that the algorithm ${\mathcal A}$ always derives a feasible
schedule. 



\subsection{Notation from Scheduling Theory}
\label{sec:scheduling-theory}

In this subsection, for completeness, we summarize the classical
flow shop and job shop scheduling problems in operations research (OR).
In scheduling theory, a scheduling problem is described by a triplet
$\alpha|\beta|\gamma$.
\begin{compactitem}
\item $\alpha$ describes the machine (i.e., processing) environment.
\item $\beta$ specifies the characteristics and constraints.
\item $\gamma$ is the objective to be optimized.
\end{compactitem}

\noindent The widely used machine environment in $\alpha$ are:
\begin{compactitem}
	\item $1$: single machine (or uniprocessor).
	\item $P$: independent machines (or homogeneous multiprocessor systems).
	\item $F_M$: \textbf{flow shop.} The environment $F_M$ consists of $M$ machines
	and each job $i$ has a chain of $M$
	sub-jobs, denoted as $O_{i,1}, O_{i,2}, \ldots, O_{i,M}$, where the
	$M$ operations are executed in the specified order and $O_{i,m}$ is executed on the $m$-th machine.
	A job has to finish its operation
	on the $m$-th machine before it can start any operation on the
	$(m+1)$-th machine, for any $m=1,2,\ldots,M-1$.
	\item $J_M$: \textbf{job shop}, i.e., a job $i$ has a chain of $\eta_i$
	sub-jobs, denoted as $O_{i,1}, O_{i,2}, \ldots, O_{i,\eta_i}$, where the
	$\eta_i$ operations should be executed in the specified order and
	$O_{i,m}$ is executed on a specified machine. Note that a flow shop
	is a special case of a job shop environment.
\end{compactitem}

\noindent In this paper, we are specifically interested in three constraints
specified in $\beta$:
\begin{compactitem}
	\item $prmp$: preemptive scheduling. In classical
	scheduling theory, preemption in parallel machines implies the
	possibility of job migration from one machine to
	another machine.
	\item $r_j$: with specified arrival time of the job (and deadline). 
	\item $l_{i,j}$: preparation time between dependent job pair, i.e., job $i$ and job $j$.
	\item $prec$: the jobs have precedence constraints.
\end{compactitem}
Note that the scheduler is implicitly assumed to be non-preemptive if
$prmp$ is not specified. Furthermore, the job set is assumed to arrive at time
$0$ if $r_j$ is not specified. 

\noindent In addition, we are specifically interested in two objectives specified in $\gamma$:
\begin{compactitem}
	\item $C_{\max}$: to minimize the makespan, as defined in Sec.~\ref{sec:problem-definition}.
	\item $L_{\max}$: to minimize the maximum lateness over all jobs, in which the
	lateness of a job is defined as its finishing time minus its
	absolute deadline.
\end{compactitem}

\subsection{Critical Sections Access Patterns}
\noindent Two types of access patterns of the critical sections are considered,
which we name according to the applicable algorithms for convenience:
\begin{compactitem}
	\item \textbf{Flow-Shop Compatible Access Patterns}: A task set has 
	a pattern where \emph{flow-shop} approaches can be applied, if all tasks access
	each resource (in a non-nested manner) at most once and a total order $\prec$ in which tasks access the resources can be constructed
	over all tasks in the set.
	Hence,
	a flow-shop pattern means that
	$\sigma_{i,j'} \prec\sigma_{i,j}$ when $j' < j$ and $\theta_{i,j'}$
	and $\theta_{i,j}$ are both critical sections. 
	In such a case, we can assume that the mutex locks
	are indexed according to the specified order. However, while tasks have to
	respect the order $\prec$  when accessing the resources, mutex locks that are
	not needed may be skipped. 
	\item 
	\textbf{Job-Shop Compatible Access Patterns}
	allow tasks to accesses shared resources multiple times
	and without any restriction on the order in which resources are accessed. 
\end{compactitem}
\emph{Flow-shop compatible access patterns} are a very restrictive special case
and of the much more general \emph{job-shop  compatible access patterns}. We 
implicitly assume job-shop compatible access patterns if not specified
differently, but examine flow-shop compatible access patterns when showing
certain complexity  results.

\section{Computational Complexity Analysis}
\label{sec:computational-complexity}

In this section, we provide a short overview of  
results regarding job shop and flow shop problems in the literature at first.
Afterwards, we explain the connection of the \tsocs{} schedulability problem to the
job and flow shop problem by showing different reductions that can be
later applied for demonstrating different scenarios with respect to their computational
complexity.  
\subsection{Literature Review of Shop Scheduling}
\label{sec:computational-complexity-shops}

Since the late 1950s, many 
computational complexity results, approximation algorithms, heuristic
algorithms, and tools for job and flow shop scheduling problems have been
established.
Intractability results have been well-established even for severely
restricted instances of job shop or flow shop problems. 
The reader is referred to the surveys by 
Lawler et al.~\cite{lawler-survey-1993} and Chen et al.~\cite{Chen1998} for details.

Specifically, the following restricted scenarios are \mbox{${\mathcal
		NP}$-complete} in the strong sense:
\begin{compactitem}
	\item $J_2||C_{\max}$, see 
	\cite{LenstraRinnooy-Kan:79:Computational-complexity}.
	\item \mbox{$J_3|p_{i,j}=1|C_{\max}$}, i.e., unit execution time, see~\cite{LenstraRinnooy-Kan:79:Computational-complexity}.
	\item \mbox{$J_3|n=3|C_{\max}$}, i.e., 3 jobs with multiple
	operations on 3 shops, 
	see~\cite{SotskovShakhlevich:95:NP-hardness-of-shop-scheduling}.
	\item $F_3||C_{\max}$, i.e., three-stage flow shop \cite{b:Gary79}.
	\item \mbox{$F_2|r_j|C_{\max}$}, i.e., two-stage flow shop with arrival times,
	as shown in~\cite{LenstraRinnooy-Kan:79:Computational-complexity}.
	\item \mbox{$F_2|p_{i,j}=1, t_j|C_{\max}$}, i.e., two-stage flow shop
	with unit processing time and transportation time between the
	finishing time of the first 
	and the starting time of the
	second stage 
	\cite{DBLP:journals/scheduling/YuHL04}.
\end{compactitem}
The best polynomial-time approximation algorithm for the general
$J_M||C_{\max}$ problem was provided by Shmoys et
al.~\cite{DBLP:journals/siamcomp/ShmoysSW94}, showing an
approximation ratio of $O\left(\frac{\log^2 (M\mu)}{\log\log (M\mu)}\right)$,
where $M$ is the number of shops and $\mu$ is the maximum number of
operations per job. The approximation ratio of this algorithm was later
improved by Goldberg et al. \cite{DBLP:journals/siamdm/GoldbergPSS01}, showing
a ratio of $O\left(\frac{\log^2 (M\mu)}{(\log\log (M\mu))^2}\right)$.

Whether there exists a polynomial-time algorithm with a constant
approximation ratio for the general $F_M||C_{\max}$ or $J_M||C_{\max}$ problem
remained 
open until 2011, when
Mastrolilli~and~Svensson~\cite{Mastrolilli:2011:HAF:2027216.2027218}
showed that $F_M||C_{\max}$ (and hence $J_M||C_{\max}$) does not
admit any polynomial-time approximation algorithm with a constant
approximation ratio. 
Moreover, they also showed that the lower bound
on the approximation ratio is very close to the existing upper bound
provided by Goldberg et al. \cite{DBLP:journals/siamdm/GoldbergPSS01}.

In Sec.~\ref{sec:complexity-small-M}, we demonstrate
that the \tsocs{} schedulability problem is already ${\mathcal
	NP}$-complete in the strong sense for very restrictive scenarios, even when $M$ and $Z$ are both extremely small. In
Sec.~\ref{sec:complexity-partitioned-big-M}, we further reduce from
the master-slave problem
\cite{DBLP:journals/scheduling/YuHL04} to
show that the \tsocs{} schedulability problem is ${\mathcal NP}$-complete
in the strong sense even when there are two critical sections that access the unique shared resource with
unit execution time per task.

\subsection{Reductions from the Job/Flow Shop Problem}
\label{sec:connection-to-jobshop}

Chen et al. \cite{Chen-Dependency-RTSS18} showed that a special case
of the \tsocs{} makespan problem is ${\mathcal NP}$-hard in the strong
sense when a task has only one critical section and $M$ is
sufficiently large. The \tsocs{} schedulability problem is the
decision version of the \tsocs{} makespan problem. We therefore
focus on the hardness of the decision version in
Definition~\ref{def:problem-decision-version-def}. 
Here, we provide reductions from the job/flow shop scheduling problems
to different restricted scenarios of the \tsocs{} schedulability
problem. 
Such reductions are used in
Sec.~\ref{sec:complexity-small-M} for demonstrating the ${\mathcal
	NP}$-completeness for different scenarios. 
We start from the more general scenario under the
semi-partitioned scheduling paradigm.

\begin{theorem}
	\label{theorem:makespan-np-hard}
	Under the semi-partitioned scheduling paradigm, there is a polynomial-time
	reduction 
	from an input instance of the decision version of
	the job shop scheduling problem $J_Z||C_{\max}$ with $Z$ shops to an
	input instance of the \tsocs{} schedulability problem
	that has $Z$ mutex locks on $M$ processors with $M \geq Z$.
\end{theorem}

\begin{proof}
	The proof is based on a polynomial-time reduction from an instance of the job
	shop scheduling problem $J_Z||C_{\max}$ to the \tsocs{} schedulability
	problem. 
	We present a polynomial-time reduction from the job
	shop scheduling problem $J_Z||C_{\max}$ to the \tsocs{} schedulability
	problem. 
	Suppose 
	a given input instance with $n$ jobs
	of the job shop scheduling problem $J_Z||C_{\max}$.  
	\begin{compactitem}
		\item We have $Z$ shops with non-preemptive execution. 
		\item A job $i$ 
		is defined by a chain of $\eta_i$ sub-jobs, denoted as
		$O_{i,1}, O_{i,2}, \ldots, O_{i,\eta_i}$. 
		The processing time
		of $O_{i,j}$ is $C_{i,j}$.
		\item These $\eta_i$ operations should be executed in the specified
		order and $O_{i,m}$ is executed on one of the given $Z$ shops,
		i.e., on shop $s(O_{i,m})$, where $s(O_{i,m}) \in \setof{1, 2,
			\ldots, Z}$.
	\end{compactitem}
	The decision version of the job shop scheduling problem is to decide
	whether there is a non-preemptive schedule whose makespan is no more
	than a given $D$.
	The polynomial-time reduction to the \tsocs{} schedulability problem
	is as follows:
	\begin{compactitem}
		\item There are $M \geq Z$ processors. 
		\item There are $Z$ mutex locks, indexed as $1,2, \ldots,Z$.
		\item For a job $i$ of the input instance of the job shop scheduling
		problem, we create a task $\tau_i$, which is composed of $\eta_i$
		computation segments. The execution time of $\theta_{i,j}$ is the
		same as the processing time of the operation $O_{i,j}$. The mutex
		lock $\sigma_{i,j}$ used by $\theta_{i,j}$ is $s(O_{i,m})$.
		\item The deadline of the tasks is $D$ and the period is $T=D$.
	\end{compactitem}
	We denote the above input 
	instance of the job shop scheduling
	problem as $I$
	(the \tsocs{} schedulability problem as $I'$, respectively). We show that there exists a feasible schedule
	$\rho$ for $I$ (in the job shop scheduling problem) if and only if
	there exists a feasible schedule $\rho'$ for $I'$ (in the \tsocs{}
	schedulability problem).\footnote{Although we do not formally define
		the schedule function of the job shop scheduling problem, we
		believe that the context is clear enough by replacing the use of
		the computation segments with the operations.}
	
	\textbf{Only-if part}: Suppose $\rho$ is a feasible schedule for $I$, i.e.,
	\begin{equation}
	\label{eq:only-if-makespan-hardness-eq1}
	\left(\sum_{m=1}^{Z} \int_{0}^{D} [
	\rho(t, m) = O_{i,j}] dt \right)= C_{i,j}, \forall O_{i,j}
	\end{equation}
	and $\rho(t,m)\neq O_{i,j}$ for any $t$ and $m$ if $s(O_{i,j}) \neq
	m$. Since the 
	execution on shops ins non-preemptive,
	if two operations $O_{i,j}$ and $O_{k, \ell}$ are supposed to be
	executed on a shop $z$, they are executed sequentially in $\rho$. As
	a result, without any conflict, for $0 \leq t \leq D$, we can set
	\begin{equation}
	\label{eq:only-if-makespan-hardness-eq2}
	\rho'(t,m) =
	\begin{cases}
	\bot & \mbox{if }\rho(t,m) = \bot\\
	\theta_{i,j} & \mbox{if }\rho(t,m) = O_{i,j}\\
	\end{cases}
	\end{equation}
	In the schedule $\rho'$, critical sections guarded by the mutex
	lock $z$ are executed sequentially on the $z$-th
	processor. Therefore, 
	\begin{equation}
	\label{eq:only-if-makespan-hardness-eq3}
	\left(\sum_{m=1}^{Z} \int_{0}^{D} [
	\rho'(t, m) = \theta_{i,j}] dt \right)= C_{i,j}, \forall
	\theta_{i,j} \in {\bf \Theta}
	\end{equation}
	and all the constraints for a feasible schedule for $I'$ are met.
	Such a schedule is a semi-partitioned and non-preemptive schedule
	(from the sub-job's perspective), which is also a global preemptive
	schedule (from the job's perspective).
	
	\textbf{If part}: Suppose that $\rho'$ is a feasible schedule for $I'$, i.e., 
	\begin{equation}
	\label{eq:if-makespan-hardness-eq1}
	\sum_{m=1}^{M} \int_{0}^{D} [
	\rho'(t, m) = \theta_{i,j}] dt = C_{i,j}, \forall \theta_{i,j} \in {\bf \Theta}   
	\end{equation}
	and the schedule $\rho'$ executes any two critical sections
	$\theta_{i,j}$ and $\theta_{k,\ell}$ with $\sigma_{i,j}=\sigma_{k,\ell}=z$ 	
	sequentially. 
	Therefore, for a mutex lock $z \in
	\setof{1,2, \ldots,Z}$, the critical sections guarded by $z$ must be
	sequentially executed.  As a result, without any conflict, for $0
	\leq t \leq D$, we can set
	\begin{equation}
	\label{eq:if-makespan-hardness-eq2}{\small
		\rho(t,z) =
		\begin{cases}
		O_{i,j} & \mbox{if } \exists m \mbox{ with }\rho'(t,m) = \theta_{i,j} \mbox{ and }
		\sigma_{i,j} = z\\
		\bot & \mbox{otherwise }\\
		\end{cases}}
	\end{equation}
	However, since we do not put any constraint on the feasible schedule
	$\rho'$, it is possible that the execution of $O_{i,j}$ on shop $z$
	is not continuous. Suppose that $a_{i,j}$ ($f_{i,j}$, respectively)
	is the first (last, respectively) time instant when $O_{i,j}$ is
	executed on shop $z$ in $\rho$. Since the schedule $\rho'$ executes
	any two critical sections $\theta_{i,j}$ and $\theta_{k,\ell}$
	sequentially when $\sigma_{i,j}=\sigma_{k,\ell}=z$, we know that for
	any $t$ between $a_{i,j}$ and $f_{i,j}$ either \mbox{$\rho(t,z) = O_{i,j}$}
	or $\rho(t,z) = \bot$. Therefore, we can simply set $\rho(t,z)$ to
	$O_{i,j}$ for any $t$ in the time interval $[a_{i,j}, a_{i,j}+C_{i,j})$ and
	set $\rho(t,z)$ to $\bot$ for any $t$ in $[a_{i,j}+C_{i,j},
	f_{i,j})$. The resulting schedule $\rho$ executes all the operations
	non-preemptively on the corresponding shops. Therefore, all the
	scheduling constraints of the job shop scheduling problem are met
	and 
	\begin{equation}
	\label{eq:if-makespan-hardness-eq3}
	\left(\sum_{m=1}^{Z} \int_{0}^{D} [
	\rho(t, m) = O_{i,j}] dt \right)= C_{i,j}, \forall O_{i,j}
	\end{equation}
	We note that there is no specific constraint of scheduling imposed
	by the schedule $\rho'$.
\end{proof}

The proof of Theorem~\ref{theorem:makespan-np-hard} is not valid for
the more restrictive partitioned scheduling paradigm, i.e., all the computation
segments of a task must be executed on the same processor, since the
constructed schedule $\rho'$ in the proof of the only-if part is not a
partitioned schedule. 
Interestingly, if we use 
an abundant number of processors, i.e., $M \geq n$, then the reduction in
Theorem~\ref{theorem:makespan-np-hard} holds for the partitioned
scheduling paradigm as well.

\begin{theorem}
	\label{theorem:makespan-np-hard-partitioned}
	Under the partitioned scheduling paradigm, there is a polynomial-time
	reduction which reduces from an input instance of the decision version of
	the job shop scheduling problem $J_Z||C_{\max}$ with $Z$ shops to an
	input instance of the \tsocs{} schedulability problem
	that has $n$ tasks and $Z$ mutex locks on $M$ processors with $M \geq n \geq Z$.
\end{theorem}
\begin{proof}
	The proof is identical to the proof of
	Theorem~\ref{theorem:makespan-np-hard} by ensuring that $\rho'$
	constructed in the only-if part in the proof of
	Theorem~\ref{theorem:makespan-np-hard} can be converted to a
	partitioned schedule. Instead of applying
	Eq.~\eqref{eq:only-if-makespan-hardness-eq2},
	since $M \geq n$, without any conflict, for $0 \leq t \leq D$ and
	$i=1,2,\ldots,n$, we can set
	\begin{equation}
	\label{eq:only-if-makespan-hardness-v2-eq2}
	\rho'(t,i) =
	\begin{cases}
	\bot &  \mbox{if } \nexists m \mbox{ with }\rho(t,m) = O_{i,j}\\
	\theta_{i,j} & \mbox{if } \exists m \mbox{ with }\rho(t,m) = O_{i,j}\\
	\end{cases}
	\end{equation}
	Since all computation segments of $\tau_i$ are executed on processor~$i$, 
	the schedule $\rho'$ is a partitioned schedule. All the
	remaining analysis follows 
	the proof of
	Theorem~\ref{theorem:makespan-np-hard}.
\end{proof}

\begin{theorem}
	\label{theorem:makespan-np-hard-basedon-FS}
	There is a polynomial-time reduction which reduces from an input
	instance of the decision version of the flow shop scheduling problem
	$F_Z||C_{\max}$ with $Z$ flow shops to an input instance of the
	\tsocs{} schedulability problem that has $Z$ mutex locks with a flow-shop
	compatible access pattern.
	The conditions in Theorems~\ref{theorem:makespan-np-hard}~and~\ref{theorem:makespan-np-hard-partitioned}
	for different scheduling paradigms with respect to constraint of
	$M$ remain the same.
\end{theorem}
\begin{proof}
	The proof is identical to the proofs of
	Theorems~\ref{theorem:makespan-np-hard}~and~\ref{theorem:makespan-np-hard-partitioned}. The
	additional condition is to access to the $Z$ mutex
	locks by following the index, starting from $1$.
\end{proof}

The above theorems show that the computational complexity of the
\tsocs{} schedulability problem is almost independent from the number
of processors (i.e., adding processors may not be helpful) and the
underlying scheduling paradigm.  The fundamental problem is the
sequencing of the critical sections.

\subsection{Computational Complexity for Small $M$}
\label{sec:complexity-small-M}

We can now reach the computational
complexity of the \tsocs{} schedulability problem when $Z \geq 2$ for
small $M$. 
For completeness,
we state the following lemma.

\begin{lemma}
	\label{lemma:makespan-in-NP}
	The \tsocs{} schedulability problem is in ${\mathcal NP}$.
\end{lemma}
\begin{proof}
	Since the feasibility of a given schedule for the \tsocs{}
	schedulability problem can be verified in polynomial-time, it is in
	${\mathcal NP}$.
\end{proof}

The following four theorems are based on the reductions in
Theorem~\ref{theorem:makespan-np-hard}~and~Theorem~\ref{theorem:makespan-np-hard-basedon-FS}.
In general, even very special cases are ${\mathcal NP}$-complete in
the strong sense.

\begin{theorem}
	\label{theorem:makespan-np-hard-Z=2}
	Under the semi-partitioned scheduling paradigm, the \tsocs{}
	schedulability problem is ${\mathcal NP}$-complete in the strong sense
	when $Z=M=2$.
\end{theorem}
\begin{proof}
	The job shop scheduling problem $J_2||C_{\max}$ with 2 shops
	is ${\mathcal NP}$-complete in the strong sense
	\cite{LenstraRinnooy-Kan:79:Computational-complexity}. Together with
	Theorem~\ref{theorem:makespan-np-hard}, we conclude the theorem.
\end{proof}

The \tsocs{} schedulability problem
is also difficult when all computation segments have the same
execution time.

\begin{theorem}
	\label{theorem:makespan-np-hard-equal-execution-time}
	Under the semi-partitioned scheduling paradigm, the \tsocs{}
	schedulability problem is ${\mathcal NP}$-complete in the strong sense
	when $Z=M=3$ and $C_{i,j}=1$ for any computation segment $\theta_{i,j}$. 
\end{theorem}
\begin{proof}
	The job shop scheduling problem \mbox{$J_3|p_{i,j}=1|C_{\max}$} with
	unit execution time on 3 shops is ${\mathcal NP}$-complete in
	the strong
	sense~\cite{LenstraRinnooy-Kan:79:Computational-complexity}.
	Together with 
	Theorem~\ref{theorem:makespan-np-hard}, we conclude the theorem.
\end{proof}

The following theorem shows that the \tsocs{} schedulability problem
is also difficult when there are just three tasks, three mutex
locks, and three processors.

\begin{theorem}
	\label{theorem:makespan-np-hard-three-tasks}
	The \tsocs{} schedulability problem is ${\mathcal NP}$-complete in the
	strong sense when $n=Z=M=3$.
\end{theorem}
\begin{proof}
	The job shop scheduling problem \mbox{$J_3|n=3|C_{\max}$} with 3
	jobs (with multiple operations) on 3 shops is ${\mathcal
		NP}$-complete in the strong
	sense~\cite{SotskovShakhlevich:95:NP-hardness-of-shop-scheduling}.
	Together with
	Theorem~\ref{theorem:makespan-np-hard}, we conclude the theorem for
	semi-partitioned scheduling paradigm.
	
	For the partitioned scheduling paradigm, since there are exactly 3
	tasks, 3 processors, and 3 mutex locks, the computational complexity
	remains the same, as a semi-partitioned schedule can be mapped to a
	partitioned schedule.
\end{proof}


\begin{theorem}
	\label{theorem:makespan-np-hard-ordered-locks}
	Under the semi-partitioned scheduling paradigm, the \tsocs{}
	schedulability problem \emph{for flow-shop compatible access patterns} is
	${\mathcal NP}$-complete in the strong sense when $Z=M=3$.
\end{theorem}
\begin{proof}
	The flow shop scheduling problem $F_3||C_{\max}$ with 3 
	shops is ${\mathcal NP}$-complete in the strong sense
	\cite{b:Gary79}.  Together with Theorem~\ref{theorem:makespan-np-hard-basedon-FS}, we conclude
	the theorem.
\end{proof}

\subsection{Computational Complexity When $M \geq N$}
\label{sec:complexity-partitioned-big-M}

Chen et al. \cite{Chen-Dependency-RTSS18} showed that a special case
of the \tsocs{} makespan problem is ${\mathcal NP}$-hard in the strong
sense when a task has only one critical section and $M$ is
sufficiently large. The following theorem shows that the \tsocs{}
schedulability problem is 
\mbox{${\mathcal NP}$-complete} when there are only two critical
sections per task and the critical sections are with unit execution
time.

\begin{theorem}
	\label{theorem:makespan-np-hard-one-lock-multiple-segments}
	The \tsocs{} schedulability problem is \mbox{${\mathcal NP}$-complete} in the
	strong sense when $Z=1$, $\eta_i \geq 3$ for every $\tau_i \in
	\textbf{T}$, $C_{i,j}=1$ for every computation segment $\theta_{i,j}$
	with $\lambda_{i,j}=1$, and $M \geq N$.
\end{theorem}
\begin{proof}
	The problem is in ${\mathcal NP}$, since the feasibility of a given schedule can
	be verified in polynomial-time.
	Similar to the proof of Theorem~\ref{theorem:makespan-np-hard}, 
	we show a polynomial-time reduction from the master-slave scheduling problem
	with unit execution time on the master 
	\cite{DBLP:journals/scheduling/YuHL04}.
	Assume a given input instance with $n$ jobs of the
	master-slave scheduling problem:
	\begin{compactitem}
		\item 
		We assume a sufficient number of slaves, but only
		one master that can be modeled as a uniprocessor.  
		\item A job $i$ has a chain of three sub-jobs, in which the first and
		third sub-jobs have to be executed on the master and the second
		sub-job has to be executed on a slave.
		\item The processing time of the first and third sub-jobs of a job
		$i$ is $1$.  The processing time of the second sub-job of a job $i$
		is $O_i > 0$.
	\end{compactitem}
	The decision version of the master-slave scheduling problem is to
	decide whether there is a schedule whose makespan is no more than a
	given target $D$, which is ${\mathcal NP}$-complete in the strong sense
	\cite{DBLP:journals/scheduling/YuHL04}. The master-slave scheduling
	problem is equivalent to the uniprocessor self-suspension problem
	with two computation segments and one suspension interval.
	
	\noindent The polynomial-time reduction to the \tsocs{} schedulability problem
	is as follows:
	\begin{compactitem}
		\item There are $M \geq n$ processors.
		\item There is one mutex lock.
		\item For a job $i$ of the input instance of the master-slave
		scheduling problem, we create a task~$\tau_i$, which is composed
		of three computation segments. The execution time
		$C_{i,1}=C_{i,3}$ and $C_{i,2} = O_i$. Computation segments
		$\theta_{i,1}$ and $\theta_{i,3}$ are critical sections guarded by
		the only mutex lock. Computation segment $\theta_{i,2}$ is a
		non-critical section.
		\item The deadline of the tasks is $D$ and the period is $T=D$.
	\end{compactitem}
	It is not difficult to prove that  a feasible schedule
	$\rho$ for the original input of the master-slave scheduling problem exists
	if and only if there exists a feasible schedule $\rho'$ for the
	reduced input of the \tsocs{} schedulability problem. Details
	are omitted due to space limitation.
\end{proof}

%
%

\section{The DGA Based on Job/Flow Shop}
\label{sec:our-algorithm-based-on-shop}
In this section,  we 
detail the DGA for tasks with multiple critical
sections, based on job shop scheduling to construct a dependency graph. 
\begin{compactitem}
	\item In
	the first step, we construct a directed \emph{acyclic} graph $G = (V, E)$. For
	each sub-job $\theta_{i,j}$ of task $\tau_i$ in $\textbf{T}$, we create a vertex in $V$. The sub-job $\theta_{i,j}$
	is a predecesor of $\theta_{i,j+1}$ for
	$j=1,2,\ldots,\eta_i-1$. Suppose that ${\bf \Theta}^z$ is the set of
	the computation segments that are critical sections guarded by mutex
	lock $z$, i.e., \mbox{${\bf \Theta}^z \leftarrow
		\set{\theta_{i,j}}{\lambda_{i,j}=1 \mbox{ and }
			\sigma_{i,j}=z}$.} For each \mbox{$z = 1, 2, \ldots, Z$,} the subgraph of
	the computation segments in ${\bf \Theta}^z$ is a directed chain,
	which represents the total 
	execution order of these
	computation segments.
	\item In the second step, we construct a schedule of $G$ on $M$
	processors either globally or partitioned, either preemptive or
	non-preemptive.
\end{compactitem}
For a directed acyclic graph $G$, a \textbf{critical path} of $G$ is a
longest path of $G$, and its length
is denoted by $len(G)$.  We now explain how to reduce from an input instance
$I^{MS}$ of the \tsocs{} makespan problem to an input instance
$I^{JS}$ of the job shop scheduling problem $J_{Z+n}||C_{\max}$.
\begin{compactitem}
	\item We create $Z+n$ shops:
	\begin{compactitem}
		\item Shop 
		$z \in \{1, 2,\ldots, Z\}$
		is exclusively used to execute critical
		sections guarded by mutex lock $z$. That is, only critical
		sections $\theta_{i,j}$ with $\lambda_{i,j}=1$ and
		$\sigma_{i,j}=z$ (i.e., $\theta_{i,j} \in {\bf \Theta}^z$) can be
		executed on shop $z$.
		\item Shop $Z+i$ is exclusively used to execute non-critical
		sections of task $\tau_i$.  That is, only non-critical sections
		$\theta_{i,j}$ with $\lambda_{i,j}=0$ can be executed on shop
		$Z+i$.
	\end{compactitem}
	\item The operation of each computation segment $\theta_{i,j}$ is 
	transformed to the corresponding shop, and the processing time is the same as the
	segment's execution time, i.e., $C_{i,j}$.
\end{compactitem}
Suppose that $\rho^{JS}$ is a feasible job shop schedule for
$I^{JS}$. Since $\rho^{JS}$ is non-preemptive, the operations on a 
shop are executed sequentially in $\rho^{JS}$. The construction of the
dependency graph $G$ sets the precedence constraints of ${\bf
	\Theta}^z$ by following the total order of the execution of the
operations on shop $z$, i.e., the shop dedicated for
${\bf \Theta}^z$ in $\rho^{JS}$.

Once the dependency graph $G$ is constructed, a
schedule $\rho^{MS}$ of the original input instance $I^{MS}$ can
be generated by applying any scheduling algorithms to schedule $G$, as
already detailed in
\cite{Chen-Dependency-RTSS18,Shi-Dependency-RTAS2019}.
Specifically, for semi-partitioned scheduling, the LIST-EDF in~\cite{Shi-Dependency-RTAS2019}
based on classical list
scheduling by Graham~\cite{DBLP:journals/siamam/Graham69} can be
applied, i.e., whenever a processor idles and 
at least one sub-job is eligible,
the sub-job with the earliest deadline starts its execution on the processor. 
Additionally, its partitioned extension
in~\cite{Shi-Partitioned-RTCSA2019}  
(P-EDF) can be applied  
to generate the partitioned schedule. 

We assume each computation
segment/sub-task executes exactly its WCET for all the releases, i.e., early completion is forbidden, thus the schedule generated
for one hyper-period is static and repeated periodically.
Accordingly, an exact schedulability test is performed by simply evaluating
the LIST-EDF or P-EDF schedule over one hyper-period to check whether
there is any deadline miss.
Since the schedule is static and repeated periodically, there is no
dynamics that can lead to the multiprocessor 
anomalies pointed out by Graham~\cite{DBLP:journals/siamam/Graham69}.

\begin{figure}[t]
	\includegraphics[width=1\linewidth]{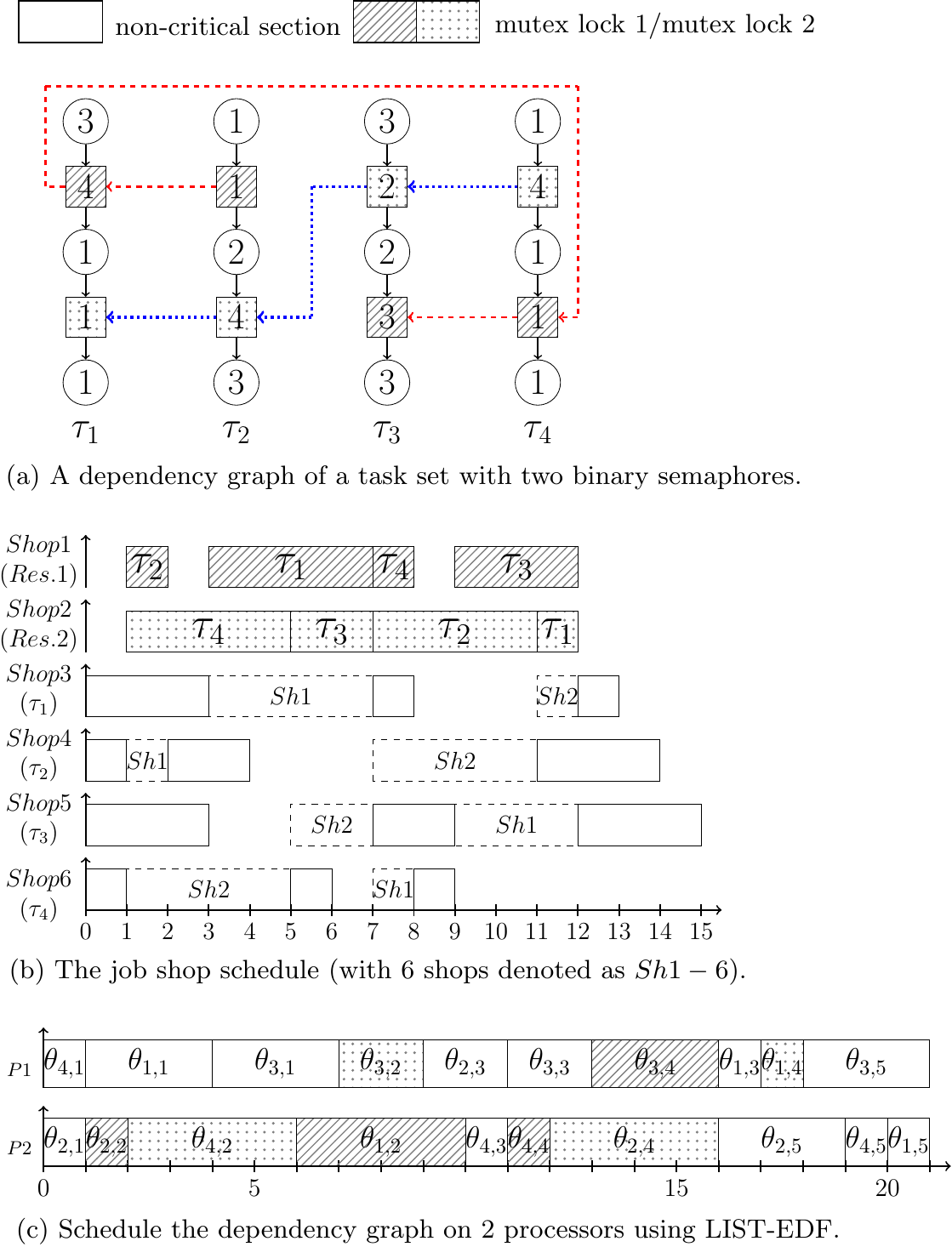}
	\caption{An example of the DGA based on job shop scheduling.}
	\label{fig:jobshop-example}
\end{figure}

\subsection{An Example of the DGA Based on Job Shop Scheduling}

To demonstrate the work flow of our approach, we
provide an illustrative example in
Fig.~\ref{fig:jobshop-example}.  Consider a frame-based task set consisting of four tasks and two shared
resources, 
where all tasks have the same period, i.e., $T_i=25$.
Each task consists of five segments, two critical sections (rectangles in
Fig.~\ref{fig:jobshop-example}~(a)) and three non-critical sections (circles).
These computation segments within one task have to be executed sequentially by following the pre-defined order (black solid
arrows in Fig.~\ref{fig:jobshop-example}~(a)).  Each of the critical
sections accesses one of the shared resources, protected
by mutex locks respectively. 
The numbers in the
circles and rectangles are the execution times of corresponding computation
segments. 

To construct a dependency graph for the task set, we apply 
job shop scheduling with~$6$ exclusively assigned 
shops: shop~$1$ and shop~$2$ are for the
critical sections of the two shared resources, and shops 3 to 6 are for
the non-critical sections of tasks $\tau_1$ to $\tau_4$. 
Hence, once a task needs to access the shared resource,
the execution will be migrated to the corresponding shops, e.g., to shop
1 for resource 1. 
This input instance
for $J_{Z+n}||C_{\max}$ is $I^{JS}$.

Fig.~\ref{fig:jobshop-example}~(b) shows a job shop schedule for $I^{JS}$.
The execution order for shared resources 1 and 2 in shop 1 and shop 2 is
according to the precedence constraints in Fig.~\ref{fig:jobshop-example}~(a), where
dashed red directed 
edges represent the precedence
constraints of mutex lock 1 and the dotted blue directed 
edges are the precedence
constraints of mutex lock~2.

The concrete schedule
is shown in Fig.~\ref{fig:jobshop-example}~(c), where the LIST-EDF presented in~\cite{Shi-Dependency-RTAS2019} is adopted
to generate the schedule on two processors. 

\subsection{Properties of Our Approach}

We now prove the equivalence of a schedule of $I^{JS}$ and a directed
acyclic graph $G$ for $I^{MS}$.

\begin{lemma}
	\label{lemma:graph-to-job-shop-schedule}
	Suppose that there is a directed acyclic graph $G$ for $I^{MS}$
	whose critical path length is $len(G)$. There is a job shop schedule
	for $I^{JS}$ whose makespan is $len(G)$.
\end{lemma}
\begin{proof}
	This lemma is proved by constructing a job shop schedule $\rho^{JS}$
	for $I^{JS}$, in which the makespan of $\rho^{JS}$ is $len(G)$.
	Suppose that the longest path 
	ended at a vertex
	$\theta_{i,j}$ in $V$ in the directed acyclic graph $G$ is
	$L_{i,j}$. There are two cases to schedule $\theta_{i,j}$ in $\rho^{JS}$:
	\begin{compactitem}
		\item If $\theta_{i,j}$ is a non-critical section, the schedule
		$\rho^{JS}$ schedules the operation 
		on shop
		$i+Z$ from time $L_{i,j}-C_{i,j}$ to $L_{i,j}$.
		\item If $\theta_{i,j}$ is a critical section guarded by mutex lock
		$z$, the schedule $\rho^{JS}$ schedules the operation 
		on shop $z$ from time $L_{i,j}-C_{i,j}$ to
		$L_{i,j}$.
	\end{compactitem}
	The above schedule has a makespan of $len(G)$ by construction. The
	only thing that has to be proved is that the schedule is a feasible
	job shop schedule for $I^{JS}$. 
	
	Suppose for contradiction that the schedule $\rho^{JS}$ is not a
	feasible job shop schedule for $I^{JS}$. This is only possible if
	the schedule $\rho^{JS}$ has a conflicting decision to schedule two
	operations at the same time $t$ on a shop $z$. There are two cases:
	\begin{compactenum}
		\item $z$ is an exclusively reserved shop for the
		non-critical sections of a task. This contradicts to the
		definition of $G$ since the non-critical sections of task $\tau_i$
		form a total order in  graph $G$.
		\item $z$ is a shop for the critical sections guarded by the
		mutex lock $z$. This contradicts to the definition of $G$ since
		the critical sections in ${\bf \Theta}^z$ 
		form a total order in graph $G$.
	\end{compactenum}
	In both cases, we reach the contradiction. Therefore, $I^{JS}$ is a
	feasible job shop schedule with a makespan of $len(G)$.
\end{proof}

\begin{lemma}
	\label{lemma:job-shop-schedule-to-graph}
	Suppose that there is a job shop schedule for $I^{JS}$ whose
	makespan is $\Delta$. Then, there is a directed acyclic graph $G$
	for $I^{MS}$ whose critical path length is at most $\Delta$.
\end{lemma}
\begin{proof}
	This lemma is proved by constructing a graph~$G$ for $I$, in which
	the critical path length of $G$ is at most~$\Delta$. By the
	definition of $G$, the sub-job $\theta_{i,j}$ is a predecesor of
	$\theta_{i,j+1}$ for $j=1,2,\ldots,\eta_i-1$ for every task
	$\tau_i$. For the sub-jobs in ${\bf \Theta}^z$, we define their total
	order and form a chain in~$G$ by following the execution order on
	shop $z$ in the given schedule~$\rho^{JS}$ for~$I^{JS}$.  Such a
	graph $G$ must be acyclic; otherwise, the schedule $\rho^{JS}$ is
	not a valid job shop schedule for $I^{JS}$.
	
	We now prove that the critical path length $len(G)$ of $G$ is no
	more than $\Delta$. Suppose for contradiction that $len(G) >
	\Delta$. This critical path of $G$ defines a total order of the
	execution of the computation segments in the critical path, which
	follows \emph{exactly} the total order of the operations of a job
	and a shop in $\rho^{JS}$. Therefore, this contradicts to the fact
	that the makespan of schedule $\rho^{JS}$ for $I^{JS}$ is $\Delta$.
\end{proof}

Based on
Lemmas~\ref{lemma:graph-to-job-shop-schedule}~and~\ref{lemma:job-shop-schedule-to-graph},
we get the following theorem:

\begin{theorem}
	\label{thm:equivalent-delivery-time+graph}
	An $a$-approximation algorithm for the job shop scheduling problem
	$J_{Z+n}||C_{\max}$ can be used to construct a dependency graph $G$
	with $len(G) \leq a \times len(G^*)$, where $G^*$ is a dependency
	graph that has the shortest critical path length for the input
	instance $I^{MS}$ of the \tsocs{} makespan problem.
\end{theorem}
\begin{proof}
	Suppose that $\Delta^*$ is the optimal makespan for $I^{JS}$. By Lemma~\ref{lemma:graph-to-job-shop-schedule}, we know that 
	\mbox{$\Delta^*
		\leq len(G^*)$}. By
	Lemma~\ref{lemma:job-shop-schedule-to-graph}, we know that $\Delta^* \geq
	len(G^*)$. Therefore, \mbox{$\Delta^* = len(G^*)$.}  Suppose that the
	algorithm derives a solution for $I^{JS}$ with a makespan~$\Delta$. By the
	\mbox{$a$-approximation} for $I^{JS}$ and
	Lemma~\ref{lemma:job-shop-schedule-to-graph}, we know $\Delta \leq a \times
	\Delta^*$. Therefore, by Lemma~\ref{lemma:job-shop-schedule-to-graph} and above discussions,
	$len(G) \leq \Delta \leq a
	\Delta^* = a \times len(G^*)$.
\end{proof}

\begin{lemma}
	\label{lemma:lower-bound}
	Let $G^*$ be defined as in
	Theorem~\ref{thm:equivalent-delivery-time+graph}.
	The optimal makespan for the input instance $I^{MS}$ of the
	\tsocs{} makespan problem is at least
	\begin{equation}\label{eq:lb}
	\max\left\{\sum_{\tau_i \in \textbf{T}} \frac{C_i}{M},
	len(G^*)\right\}
	\end{equation}
\end{lemma}
\begin{proof}
	The lower bound $\sum_{\tau_i \in \textbf{T}} \frac{C_i}{M}$ is due
	to the pigeon hole principle. The lower bound $len(G^*)$ is due to
	the definition with an infinite number of processors.  
\end{proof}

\begin{theorem}
	\label{theorem:final-approximation}
	Applying list scheduling for the dependency graph $G$ with $len(G)
	\leq a \times len(G^*)$ results in a schedule with an approximation
	ratio of $a+1$ for the \tsocs{} makespan problem under
	semi-partitioned scheduling, where $G^*$ is defined in
	Theorem~\ref{thm:equivalent-delivery-time+graph}.
\end{theorem}
\begin{proof}
	According to Theorem 1 and Section 4 in
	\cite{DBLP:journals/siamam/Graham69}, by applying list scheduling,
	the makespan of $I^{MS}$ for the \tsocs{} makespan problem is at
	most
	{\small\begin{align*}
		&len(G) + \sum_{\tau_i \in \textbf{T}} \frac{C_i}{M} 
		\leq a \times len(G^*) + \sum_{\tau_i \in \textbf{T}} \frac{C_i}{M} \\
		&\leq (a+1)
		\times \max\left\{\sum_{\tau_i \in \textbf{T}} \frac{C_i}{M},
		len(G^*)\right\}  
		\end{align*}}
	The resulting schedule is a semi-partitioned schedule since two
	computation segments of a task can be executed on different
	processors.
	By Lemma~\ref{lemma:lower-bound}, we conclude the
	theorem.
\end{proof}

Since the
\mbox{1950s~\cite{lawler-survey-1993, Chen1998}}, job/flow shop scheduling problems have been extensively studied. 
Although the problems are ${\mathcal NP}$-complete in the strong
sense (even for very restrictive cases), algorithms with different
properties have been reported in the literature. If time complexity is
not a major concern, applying constraint programming as well as mixed
integer linear programming (MILP) or branch-and-bound heuristics can
derive optimal solutions for the job shop scheduling problem. In such
a case, based on Theorem~\ref{theorem:final-approximation}, our
DGA has an approximation ratio of $2$ for the \tsocs{}
makespan problem.

\subsection{Remarks}
\label{subsec:remarks}

At first glance, it may seem  impractical to reduce the \tsocs{}
makespan problem to another very challenging problem, i.e., job shop scheduling,
in the first step of our DGA algorithms. However, 
an advantage of considering the job shop scheduling problem is that it has been
extensively studied in the literature, related results can directly be applied,
and commercial tools, like the Google {OR-Tools}~\footnote{ https://developers.google.com/optimization/},
can be utilized, as we did in our evaluation. In addition,  due to Lemma~\ref{lemma:graph-to-job-shop-schedule}, constructing a good
dependency graph implies a good schedule for $I^{JS}$.


The last $n$ job shops, i.e., shops \mbox{$Z+1, Z+2, \ldots, Z+n$}, in
$I^{JS}$, are in fact created just to match the original job shop
scheduling problem. 
From the literature of flow and job shop scheduling,
we know that these additional $n$ job shops can be removed by
introducing \emph{delay} ($l_{i,j}$ in Sec.~\ref{sec:scheduling-theory}). If the first computation segment
$\theta_{i,1}$ of task~$\tau_i$ is a non-critical section, this
implies a non-zero release time $r_i$ of task $\tau_i$ in $I^{JS}$. 

In our Google \mbox{OR-Tools}
implementation for solving $I^{JS}$, 
the no overlap constraint has to be taken into consideration
for both machine and job perspectives.
For each machine, it prevents
jobs assigned on the same machine from overlapping in time.
For each job, it prevents sub-jobs for the same job from
overlapping in time.
The first constraint 
can be achieved by
applying the \texttt{AddNoOverlap} method, by default supported in
Google OR-Tools, for each machine.
For the second constraint,
instead of
creating $n+Z$ job shops, we utilize the above concept by creating
only $Z$ job shops and adding proper delays between the
operations.
We 
configure the
start time (denoted as $\theta_{i,j}.start$) of a computation segment
based on the end time (denoted as $\theta_{i,j}.end$) of an earlier
computation segment. For notational brevity, we assign
$\theta_{i,1}.start \geq 0$ and $\theta_{i,0}.end = 0$. For any $j
\geq 2$ with $\lambda_{i,j}=1$:
\begin{equation}
\label{eq:ortools-constraints}
\begin{cases}
\theta_{i,j}.start \geq \theta_{i,j-1}.end & \mbox{ if } \lambda_{i,j-1}
\mbox{ is } 1\\
\theta_{i,j}.start \geq \theta_{i,j-2}.end  + C_{i,j-1} & \mbox{ if } \lambda_{i,j-1}
\mbox{ is } 0\\
\end{cases}
\end{equation}
In other words, if $\theta_{i,j-1}$ is a non-critical section, the
execution time $C_{i,j-1}$ is added directly to the end (finishing)
time of $\theta_{i,j-2}$; otherwise $\theta_{i,j}$
is started after the end time of $\theta_{i,j-1}$.


Hence, a proper job shop scheduling problem for $I^{JS}$ is 
$J_Z|r_j,l_j|C_{\max}$, i.e., scheduling of jobs with
release time and delays between operations on $Z$ shops.  An
$a$-approximation algorithm for the problem $J_Z|r_j,l_j|C_{\max}$ can
be used to construct a dependency graph. This problem is not
widely studied and only few results can be found in the literature.

For a task system with a \emph{flow-shop compatible access pattern}, 
i.e., the $Z$ mutex locks have a pre-defined total order,
the instance 
$I^{JS}$ is in fact a flow
shop problem. For a special case with three computation segments per
task in which the second segment is a non-critical section, and the
first and the third segments are critical sections of mutex locks $1$
and $2$, respectively, the constructed input 
$I^{JS}$ is 
a two-stage flow shop problem with delays,
i.e., $F_2|l_j|C_{\max}$. For 
the problem $F_2|l_j|C_{\max}$, several polynomial-time approximation algorithms
are known:
Karuno and Nagamochi~\cite{DBLP:conf/isaac/KarunoN03} developed a
\mbox{$\frac{11}{6}$-approximation}, Ageev~\cite{DBLP:conf/waoa/Ageev07}
developed a $1.5$ approximation for a special case when $C_{i,1}=C_{i,3}$ for
every task $\tau_i$, and
Zhang~and~van~de~Velde~\cite{DBLP:journals/scheduling/ZhangV10}
proposed polynomial-time approximation schemes (PTASes), i.e.,
$(1+\epsilon)$-approximation for any $\epsilon > 0$.

Specifically,
Zhang~and~van~de~Velde~\cite{DBLP:journals/scheduling/ZhangV10}
presented PTASes for different settings of the job/flow shop
scheduling problems 
in~\cite{DBLP:journals/scheduling/ZhangV10}. For any of such
scenarios, the approximation ratio of DGA is at most $2+\epsilon$
for any $\epsilon >0$, according to
Theorem~\ref{theorem:final-approximation}.

\subsection{Extension to Periodic Tasks}
\label{sec:extention-to-periodic}

The treatment used in \cite{Shi-Dependency-RTAS2019} 
to construct dependency graphs
can also be applied here. That is, unroll the jobs of all the tasks in one hyper-period and then construct a dependency graph of
these jobs. 
Since 
the jobs for one task 
should not have any execution overlap with each other, we only need
one dedicated shop for them. Therefore, there are two modifications of
the job shop problem scheduling considered in
Sec.~\ref{sec:our-algorithm-based-on-shop} (the studied
problem is $J_{Z+n}|r_j, l_{i,j}|L_{\max}$):
\begin{compactitem}
	\item For the $\ell$-th job, we set its release time to $(\ell-1)T_i$ 
	and its absolute deadline to $(\ell-1)T_i+D_i$.
	\item Instead of optimizing the makespan, the objective is to minimize the maximum lateness.
\end{compactitem}

In the end, the schedules are generated offline by applying
LIST-EDF or P-EDF, similar to fame-based task systems.



\section{Implementation and Overheads}
\label{sec:implementation-overhead}

In this section, we present details on how we implemented the
dependency graph approach in \litmus{} to support multiple critical sections
per task.  
Afterwards, the implementation overheads are compared with the
Flexible Multiprocessor Locking Protocol (FMLP)~\cite{block-2007}
provided by \litmus{}
for both partitioned and global scheduling.

\subsection{Implementation Details}

When implementing our approach in \litmus{}, we can either apply the table-driven scheduling that
\litmus{} provides, or implement a new binary
semaphore which enforces the execution order of critical sections that access
the same resource, since this order 
is defined in advance by the dependency graph.
A static scheduling table can be generated over one hyper-period
and be repeated periodically in a table-driven schedule.
This table determines which sub-job
is executed on which processor for each time point in the hyper-period.
However, due to the large number of sub-jobs in one hyperperiod and possible
migrations among processors, the resulting table can be very large.
To avoid this problem, we decided to implement a new binary semaphore that
supports all the properties of our new approach instead. 

Since our approach 
is an extension of the DGA by
Chen~et~al.~\cite{Chen-Dependency-RTSS18}, and
Shi~et~al.~\cite{Shi-Dependency-RTAS2019}, 
our implementation is based on the source code the authors provided
online~\cite{HDGALITMUS}, i.e., 
it is implemented under the plug-in 
Partitioned EDF with synchronization support (PSN-EDF), called P-DGA-JS, and the
plug-in Global EDF with synchronization support (GSN-EDF), 
denoted G-DGA-JS. 

The EDF feature is guaranteed by the original
design of these two plug-ins. Therefore, we only need 
to provide the relative deadlines for all the sub-jobs of each task,
and \litmus{} will automatically update the absolute deadlines accordingly
during runtime.

In order to enforce the sub-jobs to follow the execution order
determined by the dependency graph, our implementation has to:
1) let the all the sub-jobs inside one job follow the predefined order;
2) force all the sub-jobs that access the same resource to follow the 
order determined by the graph.

The first order is ensured  in \litmus{} by default.
The task deploy tool \texttt{rtspin} provided by the user-space library \emph{liblitmus}
defines the task structure, e.g., the execution order of non-critical
sections and critical sections within one task, 
the related execution times, and the resource ID that each
critical section accesses.
Moreover, the resource ID for each critical section is parsed 
by \texttt{rtspin}, so the critical section can find the correct
semaphore to lock, and in our implementation we do not
have to further consider addressing the corresponding resources.
Afterwards, \texttt{rtspin} emulates the work load in a CPU according to 
the taskset. A sub-job can be released only when its predecessor (if any)
has finished its execution.
Please note that for sub-jobs related to critical sections 
the release time is not only defined by its predecessor's finish time inside 
the same job, but also related to another predecessor that accesses
the same resource (if one exists).

A ticket system with a similar general concept to \cite{HDGALITMUS} 
is applied to enforce the execution order.
However, due to different task structure which allows to support multiple
critical sections, compared to \cite{HDGALITMUS}, additional parameters had to
be introduced and the structure of existed parameters had to be revised.
To be precise, we extended \litmus{} data structure \texttt{rt\_params} that
describes tasks, e.g., priority, period, and execution time, by adding:
\begin{compactitem}
	\item \texttt{total\_jobs}: an integer which defines the number of jobs of the 
	related task in one
	hyper-period.
	\item \texttt{total\_cs}: an integer that defines the number of critical
	sections in this task.
	\item \texttt{job\_order}: an array which 
	defines the total order of the  
	sub-jobs related to critical sections 
	that access the same resource over one
	hyper-period. In addition, the last $Z$ elements 
	record the total number of critical sections of the taskset for each shared resource. 
	Thus, the length of the array is the number of critical sections in one
	hyper-period plus the number of total shared resources,
	i.e., \texttt{len}(\texttt{job\_order}) = \texttt{total\_jobs} $\times$ \texttt{total\_cs} + $Z$.
	\item \texttt{current\_cs}: an integer that defines the index of the current
	critical section of the task that is being executed.
	\item \texttt{relative\_ddls}: an array which 
	records the relative deadlines for all sub-jobs of one task.
\end{compactitem}
Furthermore, we implemented a new binary semaphore, named as \texttt{mdga\_semaphore}, to 
make sure the execution order of all the sub-jobs that access the same resource
follows the order specified by the dependency graph.

A semaphore has the following common components:
\begin{compactitem}
	\item \texttt{litmus\_lock} protects the semaphore structure,
	\item\texttt{semaphore\_owner} defines the current holder of the semaphore,
	and 
	\item \texttt{wait\_queue} stores all jobs waiting for this semaphore.  
\end{compactitem}
A new parameter named \texttt{serving\_ticket} is 
added to control the 
non-work conserving access pattern of the critical sections, 
i.e., a job can only  lock the semaphore and 
start its critical section if it
holds the ticket equals to the
corresponding \texttt{serving\_ticket}. 

The pseudo code in Algo.~\ref{alg:multi-cs-dga-semaphore} shows three main
functions in our implementation:  
The function \textbf{\texttt{get\_cs\_order}} 
returns the position of the sub-job in the execution order 
for all the sub-jobs that access the same shared resource 
during the run-time. 
In \litmus{}, \texttt{job\_no}  
counts the number of jobs that
one task releases. 
In order to find out the exact position of this job in one hyper-period,
we apply a modulo operation on \texttt{job\_no} and \texttt{total\_jobs}.
Since a job has multiple critical section  
and the \texttt{current\_cs} represents the position of the critical section in
a job, the index is calculated by counting the number of previous jobs' critical
sections and the \texttt{current\_cs} in this job.
After that, the value
of~\texttt{cs\_order} is searched from \texttt{job\_order} based on the
obtained index.

We provide an example with 5 tasks which share two resources.  
The four tasks, i.e., 
$\tau_{1}$, $\tau_{2}$, $\tau_{3}$, and $\tau_{4}$ 
are 
identical to Fig.~\ref{fig:jobshop-example}
and 
task $\tau_{5}$ has a period $T_5=50$ and 
the same pattern as $\tau_{4}$, i.e., it
requests resource 2 in its second segment and request resource 1 in its
forth segment. 
Hence, the hyper-period for this taskset is $50$, 
$\tau_{1}$, $\tau_{2}$, $\tau_{3}$, and $\tau_{4}$  release two jobs in one
hyper-period, and $\tau_{5}$  releases one job in one hyper-period.
The related data structure is shown
in Table~\ref{tab:task-data-structure}.
Task $\tau_1$ has the \texttt{job\_order} = [1, 3, 6, 8, 9, 9].
The first two elements, i.e., [1, 3], represents that the two critical
sections of $J_1^1$ have the execution order 1 and 3 accordingly,
the following two elements, i.e., [6, 8] denotes the execution order
for $J_1^2$'s two critical sections in one hyper-period,
and the last two elements, i.e., [9, 9] shows the number of jobs that request
the related resources. For both resource 1 and resource 2,
there are nine jobs which request the resource in one hyper-period.
Assume that the 
\texttt{job\_no} for $\tau_1$ is 13.
Line 1 in Algo.~\ref{alg:multi-cs-dga-semaphore} returns the 
\texttt{current\_jobno} which represents the 
corresponding relative
position in one hyper-period, 
i.e., the $13^{th}$ job of $\tau_1$ is the 
second job of
$\tau_1$ in the current hyper-period.
Then line~2  finds the index of corresponding critical section, i.e., 
the second critical section of the second job of $\tau_{1}$ has the index~$3$.
In the end, the corresponding execution order can be found from
\texttt{job\_order} according to line 3 in Algo.~\ref{alg:multi-cs-dga-semaphore}.
Therefore, the $13^{th}$  job of task
$\tau_1$ now has the execution order $8$ to grant access to the corresponding resource.

\begin{table}[]
	\centering
	\scalebox{1.1}{
		\begin{tabular}{|c|c|c|c|c|}
			\hline
			& total\_jobs & total\_cs & job\_order        & current\_cs \\ \hline
			$\tau_1$ & 2           & 2         & {[}1,3,6,8,9,9{]} & 1           \\ \hline
			$\tau_2$ & 2           & 2         & {[}0,2,5,7,9,9{]} & 1           \\ \hline
			$\tau_3$ & 2           & 2         & {[}1,3,6,8,9,9{]} & 0           \\ \hline
			$\tau_4$ & 2           & 2         & {[}0,2,5,7,9,9{]} & 0           \\ \hline
			$\tau_5$ & 1           & 1         & {[}4,4,9,9{]}     & 0           \\ \hline  
	\end{tabular}}
	\vskip 0.1cm
	\caption{An example of the data structure for tasks.}
	\label{tab:task-data-structure}
\end{table}

\begin{algorithm}[t]
	\caption{DGA with multi-critical sections implementation}
	\label{alg:multi-cs-dga-semaphore}
	
	\begin{algorithmic}[1]
		\small
		\INPUT New coming task $\tau_{i}$\{\texttt{job\_no}, \texttt{total\_jobs},
		\texttt{total\_cs}, \texttt{current\_cs}, \texttt{relative\_ddls}\}, and
		Requested semaphore $s_z$\{\texttt{semaphore\_owner}, 
		\texttt{serving\_ticket}, \texttt{wait\_queue}\};
		\\~\\ \textbf{Function} {get\_cs\_order}():
		\STATE \texttt{current\_jobno} $\leftarrow$ $\tau_{i}$.\texttt{job\_no}\ mod\ $\tau_{i}$.\texttt{total\_jobs};
		\STATE index $\leftarrow$ \texttt{current\_jobno}\ $\times$ \ $\tau_{i}$.\texttt{total\_cs}
		+ \texttt{current\_cs};
		\STATE \texttt{cs\_order} $\leftarrow$ $\tau_{i}$.\texttt{job\_order}[index];
		\\~\\\textbf{Function} {mdga\_lock}():
		\IF {$s_z$.\texttt{semaphore\_owner} is NULL and \\
			$s_z$.\texttt{serving\_ticket} equals to $\tau_{i}$.\texttt{cs\_order}}
		\STATE $s_z$.\texttt{semaphore\_owner} $\leftarrow$ $\tau_{i}$;
		\STATE Update the deadline for $\tau_i$;
		\STATE $\tau_{i}$ starts the execution of its critical section;
		\ELSE 
		\STATE Add $\tau_{i}$ to $s_z$.\texttt{wait\_queue};
		\ENDIF
		~\\\textbf{Function} {mdga\_unlock}():
		\STATE $\tau_{i}$ releases the semaphore lock;
		\STATE Update the deadline for $\tau_i$;
		\STATE $\tau_i$.\texttt{current\_cs}++;
		\IF {$\tau_i$.\texttt{current\_cs} $=$ \texttt{total\_cs}}
		\STATE Set $\tau_i$.\texttt{current\_cs} $\leftarrow$ 0;
		\ENDIF
		\STATE $s_z$.\texttt{serving\_ticket}++;		
		\IF {$s_z$.\texttt{serving\_ticket} $=$ \texttt{num\_cs}}
		\STATE Set $s_z$.\texttt{serving\_ticket} $\leftarrow$ 0;
		\ENDIF		
		\STATE Next task $\tau_{next}$ $\leftarrow$ the head of the \texttt{wait\_queue} (if exists);
		\IF {\texttt{serving\_ticket} equals to $\tau_{next}$.\texttt{cs\_order}}
		\STATE $s_z$.\texttt{semaphore\_owner} $\leftarrow$ $\tau_{next}$; 
		\STATE $\tau_{next}$ starts the execution of its critical section;
		\ELSE 
		\STATE $s_z$.\texttt{semaphore\_owner} $\leftarrow$ NULL;
		\STATE Add $\tau_{next}$ to $s_z$.\texttt{wait\_queue};
		\ENDIF	
		
	\end{algorithmic}
\end{algorithm}

The function \textbf{\texttt{mdga\_lock}} 
is called in order to 
lock the semaphore and get access to the corresponding resource. 
After getting the correct position in the execution order in one hyper-period
by applying function \texttt{get\_cs\_order()},
the semaphore's ownership will be checked.
If the semaphore is occupied by another job at that moment, 
the new arriving job will be added to the \texttt{wait\_queue} directly;
otherwise,
the semaphore's \texttt{current\_serving\_ticket} and
the job's \texttt{cs\_order} are compared.
If they are equal, 
the semaphore's owner will be set to that job, and
the job will start its critical section;
otherwise, the job will be added to the \texttt{wait\_queue} as well.
In our setting the \texttt{wait\_queue} is sorted by the jobs'
\texttt{cs\_order}, i.e., the job with the smallest \texttt{cs\_order} is 
the head of the waiting queue.
Hence, only the head of the \texttt{wait\_queue} has to be checked
when the current semaphore owner finishes its execution,
rather than checking the whole unsorted \texttt{wait\_queue}. 

The function \textbf{\texttt{mdga\_unlock}} 
is called once a job has finished its critical section and tries to unlock
the semaphore.
The task's \texttt{current\_cs} is added by one to point to the next possible critical
section in this job. If \texttt{current\_cs} reaches to the 
\texttt{total\_cs}, which means all the critical sections in this 
job have finished their execution, then the \texttt{current\_cs} will be reset to zero.
Next, the semaphore's \texttt{serving\_ticket} is increased by
1, i.e., it is ready to be obtained by the 
successor in the dependency graph. 
If \texttt{serving\_ticket} reaches 
the total number of critical sections related to 
this resource in one hyper-period, i.e., \texttt{num\_cs}, 
the dependency graph is traversed completely, i.e., all
sub-jobs that access the related resource finished their
executions of the critical sections in the current hyper-period, 
the parameter \texttt{serving\_ticket} is reset to
$0$ to start the next iteration.
Please note, the \texttt{num\_cs} can be found in the last $Z$ elements of
\texttt{job\_order} according to the related resource id.
After that, the first job (if any) in the \texttt{wait\_queue}, named as
$\tau_{next}$ is checked.
If $\tau_{next}$ has the \texttt{cs\_order} which equals to the
semaphore's \texttt{serving\_ticket},
the the semaphore's owner is set as $\tau_{next}$,
and $\tau_{next}$ can start the execution of its critical section.
Otherwise, the semaphore owner is set as {NULL}, and the task $\tau_{next}$ is
put back to the corresponding \texttt{wait\_queue}.

Additionally, each sub-job has its own modified deadline accordingly,
which means each job can have different deadlines when it is executing different
segments. Therefore, we have to take 
care of the deadline update during the implementation.
When we deploy a task using \texttt{rtspin} to the system, 
we deliver the relative deadline of its first sub-task as the relative deadline
of the whole task.
Since no two continuous non-critical sections are allowed in the task model,
once a sub-job finishes its execution, either \texttt{mdga\_lock}
or \texttt{mdga\_unlock} is called.
If \texttt{mdga\_lock} is called, the new critical section's deadline is updated by searching the
\texttt{relative\_deadline}; if \texttt{mdga\_lock} is called, only the finished critical section can 
update related job's deadline for its successor (if any), since $\tau_{next}$'s deadline has been updated
when it tries to lock the semaphore already.

The implementations for the global and partitioned plug-ins are similar. 
However, due to the frequent preemption and/or interrupts in global
scheduling, the preemption has to be disabled during the executions
of semaphore related functions in order to protect the
functionalities of aforementioned functions.

\subsection{Overheads Evaluations}

We evaluated the overheads of our implementation in the
following platform:
a cache-coherent SMP, consisting of two 64-bit Intel Xeon Processor
E5-2650Lv4, with 35 MB cache and 64 GB main memory. The FMLP supported in
\litmus{} was also evaluated for comparisons, including P-FMLP for
partitioned scheduling and G-FMLP for global scheduling. 
These four protocols are evaluated using same task sets where
each task has multiple critical sections.

The overheads that we tracked are: 
\begin{compactitem}
	\item \textbf{CXS}: context-switch overhead.
	\item \textbf{RELEASE}: time spent to enqueue a newly released job into a ready queue.
	\item \textbf{SCHED}: time spent to make a scheduling decision, i.e., find the next job to be executed.
	\item \textbf{SCHED2}: time spent to perform post context switch and management activities.
	\item \textbf{SEND-RESCHED}: inter-processor interrupt latency, including migrations.
\end{compactitem}

The overheads are reported in Table~\ref{tab:overheads}, which shows
that the overheads of our approach and those of 
P-FMLP, G-FMLP are
comparable. 
Furthermore, the implementations provided in
\cite{Shi-Dependency-RTAS2019}, called P-LIST-EDF and G-LIST-EDF, were
evaluated to examine the overhead and reported in
Table~\ref{tab:overheads}. The direct comparison between P-LIST-EDF
and P-DGA-JS (G-LIST-EDF and \mbox{G-DGA-JS}, respectively) is not possible
because they are designed for different scenarios, depending on the
number of critical sections per task.  The reported
overheads in Table~\ref{tab:overheads} for our approach are for task
sets with multiple critical sections per task, whilst the overheads
for P-LIST-EDF and G-LIST-EDF were for task sets with one critical
section per task. Regardless, 
they are in the same order of magnitude.

\begin{table}[]
	\centering
	\vskip 0.2cm
	\def\arraystretch{1.4}  
	\scalebox{0.7}{\begin{tabular}{|c|c|c|c|c|c|c|}
			\hline
			Max. (Avg.) in $\mu s$& CXS    & RELEASE   & SCHED   & SCHED2 & SEND-RESCHED\\ \hline 
			P-FMLP      & 29.51 (0.98)  & 17.68 (0.96) & 31.85 (1.31) & 28.77 (0.18) & 66.33 (2.86) \\ \hline
			P-DGA-JS    & 30.65 (1.25)  & 18.63 (1.02)  & 31.09 (1.64) & 29.43 (0.19) & 59.09 (21.06) \\ \hline
			G-FMLP      & 30.51 (1.05)  & 48.53 (3.75) & 45.99 (1.51) & 29.62 (0.16) & 72.26 (2.50) \\ \hline
			G-DGA-JS    & 26.87 (0.94)  & 30.01 (2.19) & 30.25 (1.02) & 19.26 (0.14) & 72.53 (21.50) \\ \Xhline{5\arrayrulewidth}			
			P-LIST-EDF  & 18.76 (0.90)  & 18.98 (1.06) & 48.50 (1.33) & 29.25 (0.16)& 38.3 (1.61) \\ \hline
			G-LIST-EDF  & 30.87 (1.79) & 61.63 (12.06) & 59.05 (4.46) & 27.17 (0.25)& 72.09 (20.77) \\ \hline
	\end{tabular}}
	\vskip 0.1cm 
	\caption{\normalsize Overheads of protocols in \litmus{}.}
	\label{tab:overheads}
\end{table}

\section{Evaluations}
\label{sec:evaluations}
We evaluated the performance of the proposed approach by
applying numerical evaluations for both frame-based task sets
and periodic task sets, and measuring its overheads.

\subsection{Evaluations Setup}
\label{sec:evaluation-setup}
We conducted evaluations on $M$ = 4, 8, and 16 processors. Based on 
$M$, we generated $100$ synthetic task sets with $10M$ tasks each, 
using 
the RandomFixedSum method~\cite{emberson2010techniques}.
We set 
$\sum_{\tau_i \in \textbf{T}} U_i=M$ and 
enforced $U_i \leq 0.5$ for each task $\tau_i$, 
where $U_i =
\frac{C_i}{T_i}$ is the utilization of a task.
The number of shared resources (binary semaphores) $Z$ was 
either $4$, $8$, or $16$. 
Each task $\tau_i$ accesses the available shared resource 
randomly between $2$ and $5$ times, 
i.e., $\sum \lambda_{i,j} \in [2, 5]$.
The total length of the critical sections $\sum_{\lambda_{i,j}=1} C_{i,j}$ 
is a fraction of the total execution
time $C_i$ of task $\tau_i$, depended on
\mbox{$H \in \{[5\%-10\%], [10\%-40\%], [40\%-50\%]\}$.}
When considering shared resources in real-time systems, 
the utilization of critical sections for each task in 
classical settings is relatively low. 
However, with the
increasing computation demand in real-time systems (e.g., for machine
learning algorithms), adopted accelerators, like GPUs, 
behave like classical shared resources (i.e., they are non-preemptive and
mutually exclusive), but have a relatively high utilization. 
Hence, we chose possible settings of $H$ that cover the complete spectrum.  
The total length of critical sections and non-critical sections are split
into dedicated segments by applying UUniFast~\cite{emberson2010techniques} separately.
For task $\tau_i$, the number of critical sections $Num_{cs}$ 
equals to $\sum \lambda_{i,j}$, 
and the number of non-critical sections $Num_{ncs}=Num_{cs}+1$.
In the end, the generated non-critical sections and critical sections
are combined in pairs, and the last segment is the last non-critical section. 
We evaluated all resulting 27 combinations of $M$, $Z$, and $H$.

The dependency graph is generated by applying:
\begin{compactenum}
	\item The method 
	in Sec.~\ref{sec:our-algorithm-based-on-shop}
	with the objective to minimize the makespan, denoted as \textbf{JS}.
	We utilized the constraint programming approach provided in the Google OR-Tools to solve the job shop scheduling problem,
	\item The extension to multiple critical sections 
	sketched in~\cite{Shi-Dependency-RTAS2019}, denoted as \textbf{PRP}. To check the feasibility of the generated dependency graph, 
	one simulated schedule with respect to the dependency graph 
	is generated. 
\end{compactenum}
We name these algorithms by combining:
\begin{compactenum}
	\item \emph{JS/PRP}: 
	the two different dependency graph
	generation methods.
	\item \emph{LEDF/PEDF}: 
	to schedule the generated graph, we used the
	LIST-EDF in~\cite{Shi-Dependency-RTAS2019} (LEDF) or 
	partitioned EDF (PEDF) 
	in~\cite{Shi-Partitioned-RTCSA2019}, and a \emph{worst-fit} partitioning
	algorithm.  
	\item \emph{P/NP}: 
	preemptive or non-preemptive schedule for critical sections. 
\end{compactenum}
\noindent We also compare our approach with the following 
protocols regarding their
schedulability by applying the publicly available tool SET-MRTS in
\cite{SET-MRTS} with the same naming:
\begin{compactitem}
	\item Resource Oriented Partitioned PCP (ROP-PCP)~\cite{RTSS2016-resource}: 
	Binds the resources on dedicated processors 
	and schedules tasks using semi-partitioned PCP.
	\item GS-MSRP~\cite{wieder-2013}: THe Greedy Slacker (GS)
	partitioning heuristic for spin-based locking protocol
	MSRP~\cite{DBLP:conf/rtss/GaiLN01}, using  Audsley's Optimal Priority
	Assignment~\cite{Audsley1991aOPA} for priority assignment.
	(LP) analysis for global FP scheduling using the FMLP~\cite{block-2007}. 
	 	\item LP-GFP-PIP: LP-based global FP scheduling using the
	 	Priority Inheritance Protocol (PIP)~\cite{DBLP:conf/rtss/EaswaranA09}.
	\item LP-PFP-DPCP~\cite{bbb-2013}: DPCP~\cite{DBLP:conf/rtss/RajkumarSL88}
	with a Worst-Fit-Decreasing (WFD) task assignment strategy~\cite{bbb-2013}. The
	analysis is based on a linear-programming (LP). 
	\item LP-PFP-MPCP~\cite{bbb-2013}:  MPCP~\cite{Rajkumar_1990}
	with a Worst-Fit-Decreasing (WFD) task assignment strategy as proposed
	in~\cite{bbb-2013}. The analysis is based on a LP.
	\item LP-GFP-FMLP~\cite{block-2007}:  FMLP~\cite{block-2007} for 
	global FP scheduling  with a LP analysis.
\end{compactitem}
\emph{Note that a comparison to the original DGA
	in~\cite{Chen-Dependency-RTSS18}
	is not possible, since the approach
	in \cite{Chen-Dependency-RTSS18}  is only applicable when there is one
	critical section per task.}
We also launched the evaluation of the Priority Inheritance Protocol
(PIP)~\cite{DBLP:conf/rtss/EaswaranA09} 
based on LP, but we were not
able to collect the complete results because validating a task set took multiple
hours. However, according to \cite{DBLP:conf/rtss/YangWB15,
	Chen-Dependency-RTSS18, Shi-Dependency-RTAS2019}, the
PIP based on LP 
performs similar to LP-GFP-FMLP.

\subsection{Evaluation Results for Frame-Based Tasks}
\label{sec:experiments-frame-based}
\begin{figure}[t]
	\centering
	\includegraphics[width=\linewidth]{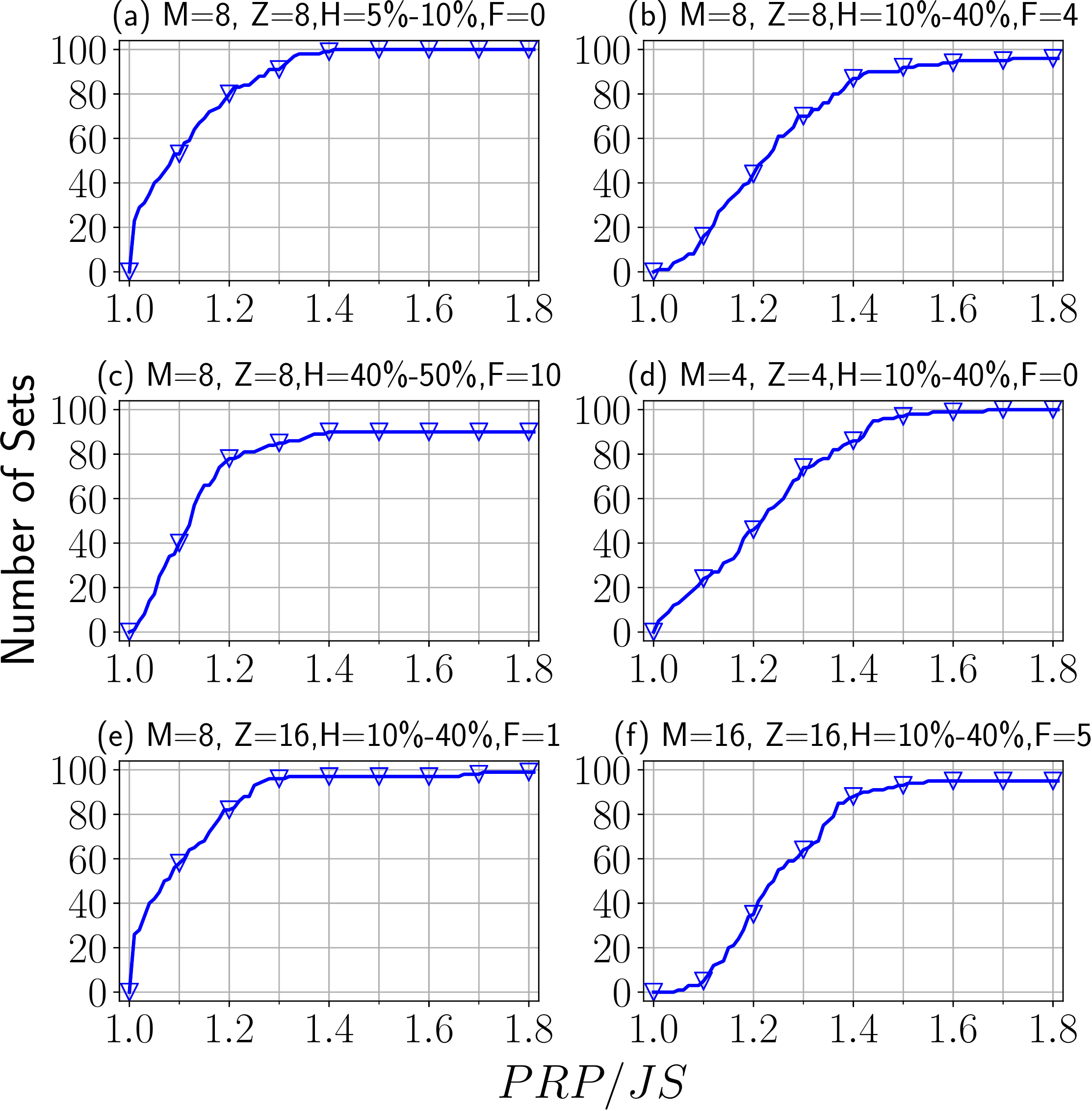}
	\caption{Comparison of critical paths from the two graph generation methods.}
	\label{fig:ratio}
\end{figure}
For frame-based task systems, we set $T=D=1$ for all the tasks, 
i.e., the
execution time of each task is the same as its utilization.
We 
tracked the number of dependency graphs
calculated with PRP where the ratio of $PRP/JS$ is less than a certain factor. 
The results are shown in
Fig.~\ref{fig:ratio}, where $F$ represents
the number of infeasible dependency graph for the  \emph{PRP} method due to
cycle detection.
The job-shop based dependency graph generation method clearly outperform the
method extended from the original DGA. 
In addition, the failure rate of the \emph{PRP} is increasing when the length
of critical sections is increased, i.e., Fig.~\ref{fig:ratio} (a), (b), and (c). The other results show similar trends and are thus omitted due to space
limitation.

In our schedulability evaluation, 
we considered synthetic task sets under the aforementioned
settings, testing the utilization level from $0$ to $100\% \times M$ 
in steps of $5\%$.
The acceptance ratios of
LP-PFP-DPCP and LP-PFP-MPCP are zero for all configurations, even for
utilization levels $\leq 20\% \times M$.
Hence, we omitted them in Fig.~\ref{fig:sched_frame}. 
Additionally, considering the readability of the figure, 
we only show 
\emph{PRP-LEDF-P}, which has 
the best performance for the approaches where dependency graphs are
generated by \emph{PRP}.

Fig.~\ref{fig:sched_frame} shows that 
our approach outperforms
the other non-DGA based methods significantly for all evaluated settings,
and performs slightly better than the methods 
using \emph{PRP}. 
Fig.~\ref{fig:ratio} and Fig.~\ref{fig:sched_frame} also show
that a better dependency graph, i.e, a shorter critical path, 
not always results in better schedulability in the second step of the DGA.



\subsection{Evaluation Results for Periodic Tasks}
\label{sec:experiments-periodic}

\begin{figure}[t]
	\centering
	\includegraphics[width=1\linewidth,trim={0 0mm 0
		0},clip]{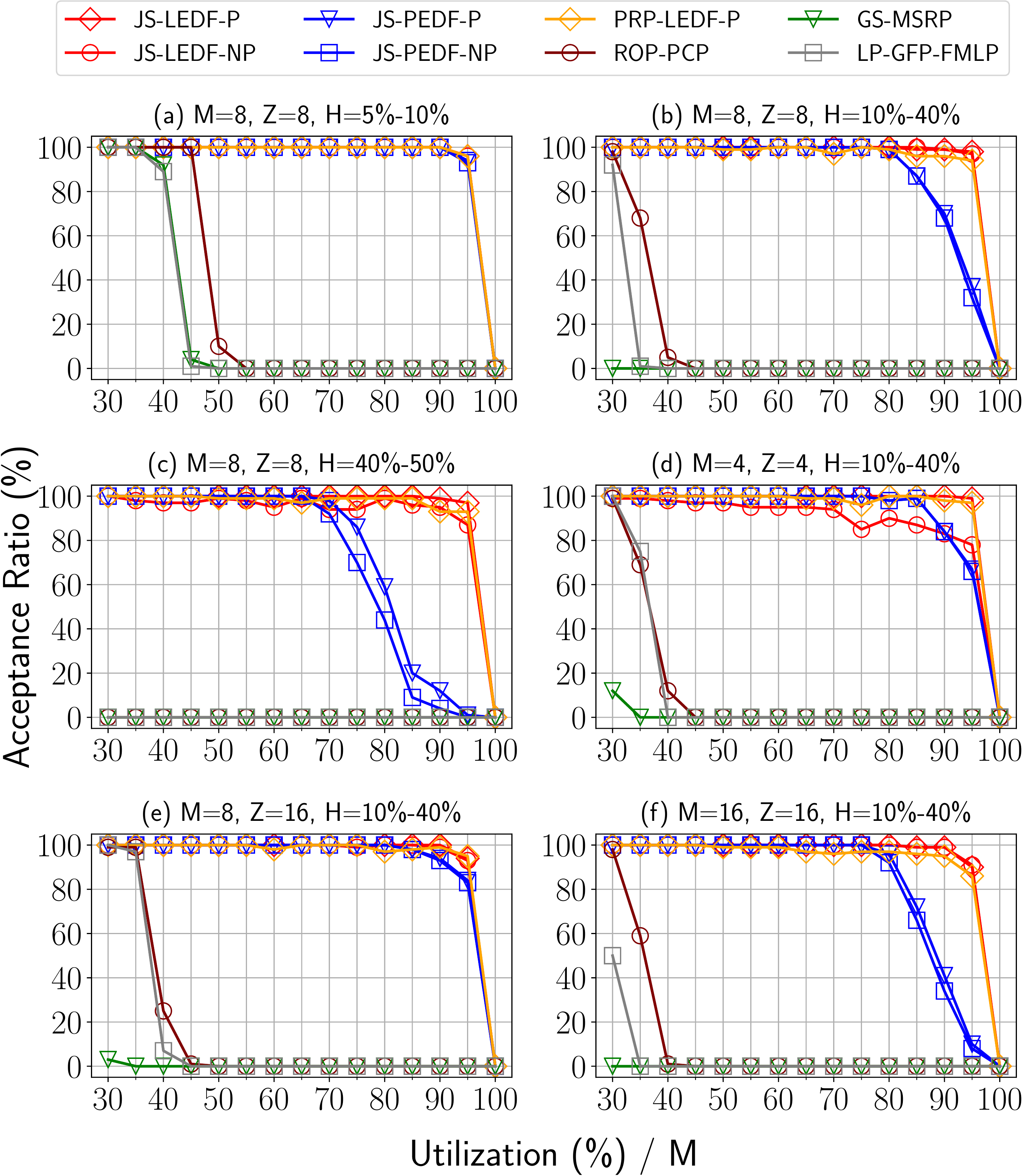}
	\caption{Schedulability of different approaches for frame-based task sets.}
	\label{fig:sched_frame}
\end{figure}

We applied constraint programming to solve the job shop
problem $J_{Z}|r_j, l_j|L_{\max}$ and construct the dependency graph. 
We extended the settings 
for frame-based task sets in Sec.~\ref{sec:experiments-frame-based}
to periodic task systems by choosing the period $T_i$ 
randomly from a set of semi-harmonic periods, i.e., $T_i\in\{1, 2, 5, 10\}$,
which is a subset of the periods used in automotive
systems~\cite{Kramer2015,DBLP:conf/ecrts/HamannD0PW17}.
We
used a small range of periods to generate reasonable task sets with
high utilization of the critical sections, which are otherwise
by default not schedulable. 

\begin{figure}[t]
	\centering
	\includegraphics[width=1\linewidth]{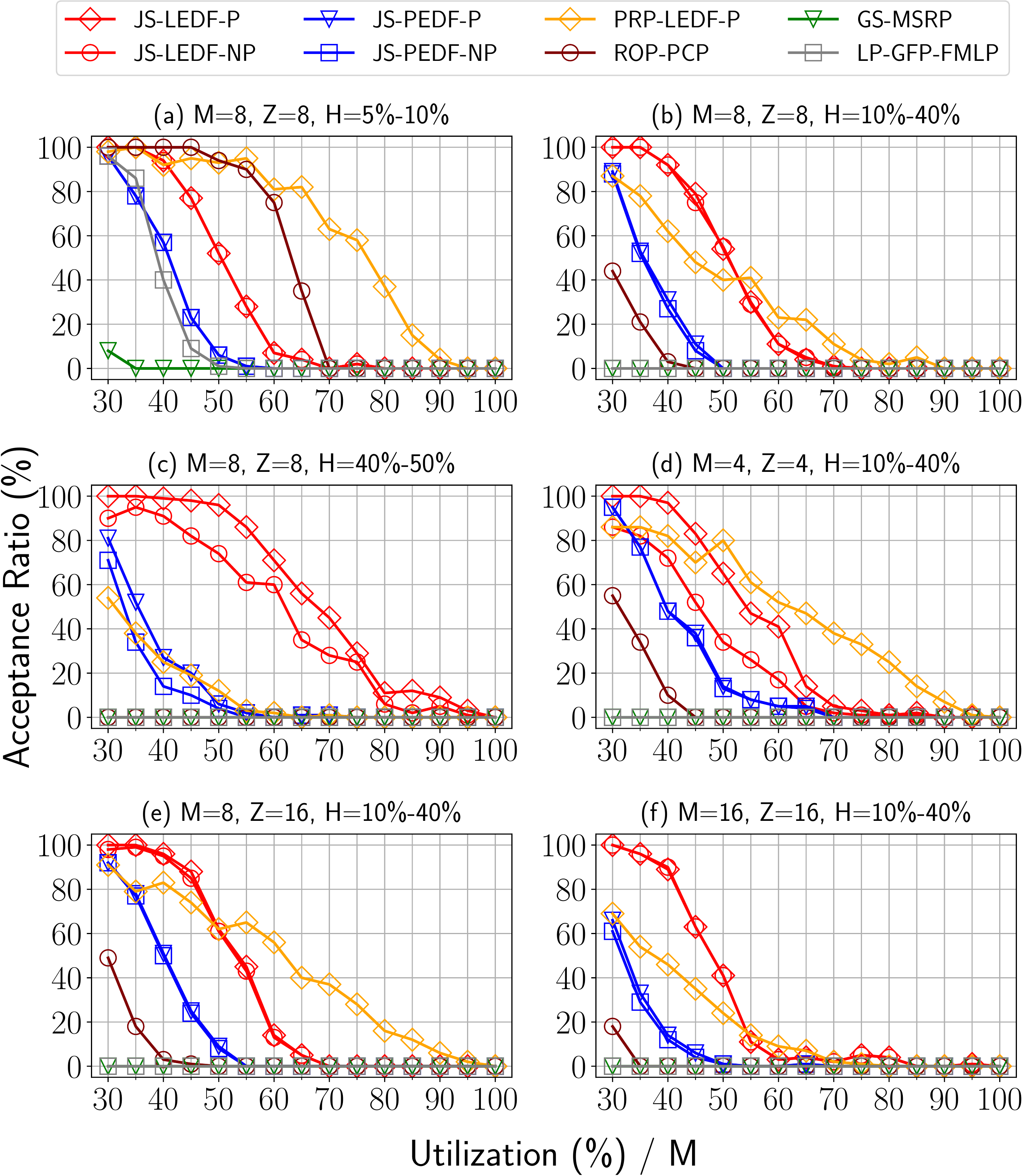}
	\caption{Schedulability of different approaches for
		periodic task sets.}
	\label{fig:sched_periodic}
\end{figure}
Due to space limitation, only a subset of the results is presented in
Fig.~\ref{fig:sched_periodic}.  When
the utilization of critical sections is high, i.e., ${H = [40\%-50\%]}$ in
Fig.~\ref{fig:sched_periodic}~(c), 
or under medium utilization when the number of processor and shared resources
are relative high, i.e., ${M = H = 16}$ in Fig.~\ref{fig:sched_periodic}~(f),
our approaches outperforms the other methods significantly. 
However, when the utilization of critical
sections is low, i.e., $H=[5\%-10\%]$
in~Fig.~\ref{fig:sched_periodic}~(a)~and~(b), ROP-PCP
outperformed the proposed approaches. The reason is that the
constraint programming of the problem $J_{Z}|r_j,l_j|L_{\max}$ has the
objective to minimize the maximum lateness, but ignores the execution
order of the sub-jobs that do not have any influence on the optimal
lateness, which may lead to lower performance when the utilization of the non-critical sections
is high.  
When the utilization of critical section is  medium, i.e.,
$H=[10\%-40\%]$, and the number of processor is relative small i.e., $M=\{4, 8\}$,
the newly proposed 
DGA-based methods 
and the extension PRP-LEDF-P both outperform all the other methods
significantly, but their relation differs 
depending on the utilization value.

\section{Conclusion}
\label{sec:conclusion}

We have removed an
important restriction, i.e., only one critical section per task, of the
recently developed dependency graph approaches (DGA). 
Regarding the computational
complexity, we show that the multiprocessor synchronization problem is
${\mathcal NP}$-complete even in very restrictive scenarios, as detailed in
Sec.~\ref{sec:computational-complexity}. We propose a
systematic design flow based on the DGA by using existing algorithms
developed for job/flow shop scheduling and provide the approximation ratio(s)
for the derived makespan. 

The evaluation results in Sec.~\ref{sec:experiments-frame-based} show
that our approach is very effective for frame-based real-time task
systems. Extensions to periodic task systems are presented
in Sec.~\ref{sec:extention-to-periodic}, and the evaluation results show that 
our approach has significant improvements, compared to existing protocols, in most
evaluated cases except light shared resource utilization. 
This paper significantly improves the applicability of the DGA by
allowing arbitrary configurations of the number of non-nested critical sections
per task.

\ifCLASSOPTIONcompsoc
  \section*{Acknowledgments}
\else
  \section*{Acknowledgment}
\fi

{This paper is supported by DFG, as  part of the
	Collaborative Research Center SFB876, project A1 and A3
	{(http://sfb876.tu-dortmund.de/)}.} 
The authors thank Zewei Chen and Maolin Yang
for their tool SET-MRTS~\cite{SET-MRTS} (Schedulability Experimental Tools for
Multiprocessors Real Time Systems) to evaluate the GS-MSRP, LP-GFP-FMLP,
and ROP-PCP in Fig.~\ref{fig:sched_frame} and
Fig.~\ref{fig:sched_periodic}.

\ifCLASSOPTIONcaptionsoff
  \newpage
\fi



%
\bibliographystyle{abbrv} \bibliography{real-time}
%

\ifArxivVersion
\else
\begin{IEEEbiographynophoto}{Jian-jia Chen}
is Professor at Department of Informatics in
TU Dortmund University, Germany. He received his Ph.D.
degree in from Department of Computer Science and Infor-
mation Engineering, National Taiwan University, Taiwan in
2006. His research interests include real-time systems, em-
bedded systems, power/energy-aware designs, and distributed
computing. 
\end{IEEEbiographynophoto}
\begin{IEEEbiographynophoto}{Junjie Shi}
received his master degree in electronic technology and information technology
from TU Dortmund University, Germany, in 2017 and now
is a PhD student at TU Dortmund University, supervised by
Prof. Dr. Jian-Jia Chen. His research interests are 
resource-sharing protocols for real-time systems,
resource aware scheduling for machine learning algorithms,
and computation offloading for real-time systems.
\end{IEEEbiographynophoto}
\begin{IEEEbiographynophoto}{Georg von der Br\"uggen}
is a Postdoctoral Researcher at the Max Planck Institute for Software Systems in Kaiserslautern, Germany. He received his PhD  
from TU Dortmund University, Germany, in 2019 and 
his Diploma degree in computer science from TU Dortmund University, Germany, in 2013.
His research interests are in the area of embedded and real-time systems with a
focus on real-time scheduling. 
He participated in the program committee of
multiple international conferences and workshops and was the program chair of the RTNS junior workshop JRWRTC in 2018.
\end{IEEEbiographynophoto}
\begin{IEEEbiographynophoto}{Niklas Ueter}
received his master degree in computer science
from TU Dortmund University, Germany, in 2018 and now
is a PhD student at TU Dortmund University, supervised by
Prof. Dr. Jian-Jia Chen. His research interests are in the area
of embedded and real-time systems with a focus on real-time
scheduling.
\end{IEEEbiographynophoto}
\fi




\end{document}